%% file: main.tex
\documentclass[11pt]{article}

\usepackage{amsthm}
\usepackage{graphicx} 
\usepackage{array} 

\usepackage{amsmath, amssymb, amsfonts, verbatim}
\usepackage{hyphenat,epsfig,subfigure,multirow}

\usepackage[usenames,dvipsnames]{xcolor}
\usepackage[ruled]{algorithm2e}

\usepackage{tcolorbox}
\tcbuselibrary{skins,breakable}
\tcbset{enhanced jigsaw}

\usepackage[compact]{titlesec}

\definecolor{DarkRed}{rgb}{0.5,0.1,0.1}
\definecolor{DarkBlue}{rgb}{0.1,0.1,0.5}

\usepackage{hyperref}
\hypersetup{
colorlinks=true,
pdfnewwindow=true,
citecolor=ForestGreen,
linkcolor=DarkRed,
filecolor=DarkRed,
urlcolor=DarkBlue
}

\usepackage{bm}
\usepackage{url}
\usepackage{xspace}
\usepackage[mathscr]{euscript}

\usepackage{mdframed}

\usepackage[noend]{algpseudocode}
\makeatletter
\def\BState{\State\hskip-\ALG@thistlm}
\makeatother

\usepackage{cite}
\usepackage{enumitem}

\usepackage[margin=1in]{geometry}

\newtheorem{theorem}{Theorem}

\newtheorem{lemma}{Lemma}[section]
\newtheorem{proposition}[lemma]{Proposition}

\newtheorem{claim}[lemma]{Claim}
\newtheorem{fact}[lemma]{Fact}

\newtheorem{definition}{Definition}

\newtheorem{problem}{Problem}

\newtheorem{remark}[lemma]{Remark}

\newtheorem*{claim*}{Claim}
\newtheorem*{proposition*}{Proposition}
\newtheorem*{lemma*}{Lemma}
\newtheorem*{problem*}{Problem}
\newtheorem*{remark*}{Remark}

\newtheorem{mdresult}{Result}
\newenvironment{result}{\begin{mdframed}[backgroundcolor=lightgray!40,topline=false,rightline=false,leftline=false,bottomline=false,innertopmargin=2pt]\begin{mdresult}}{\end{mdresult}\end{mdframed}}

\newtheorem{mdinvariant}{Invariant}

\theoremstyle{definition}
\newtheorem{mdexample}{Example}[section]
\newenvironment{example}{\begin{mdframed}[hidealllines=false,backgroundcolor=gray!10,innertopmargin=0pt]\begin{mdexample}}{\end{mdexample}\end{mdframed}}

\allowdisplaybreaks

\DeclareMathOperator*{\argmax}{arg\,max}

\renewcommand{\qed}{\nobreak \ifvmode \relax \else
      \ifdim\lastskip<1.5em \hskip-\lastskip
      \hskip1.5em plus0em minus0.5em \fi \nobreak
      \vrule height0.75em width0.5em depth0.25em\fi}

\newcommand{\Qed}[1]{\ensuremath{\qed_{\textnormal{~#1}}}}

\newcommand*\samethanks[1][\value{footnote}]{\footnotemark[#1]}

\input{macros}

\title{Stochastic Submodular Cover with Limited Adaptivity}
\author{Arpit Agarwal \and Sepehr Assadi\thanks{Supported in part by the National Science Foundation grant CCF-1617851.} \and Sanjeev Khanna\samethanks}
\date{\vspace{-10pt} University of Pennsylvania \\ \texttt{\{aarpit,sassadi,sanjeev\}@cis.upenn.edu}}

\begin{document}
\maketitle

\thispagestyle{empty}
\input{abstract}
\setcounter{page}{0}
\clearpage

\input{intro}

\input{prelim}

\input{technical}

\input{non-adaptive-selection}

\input{r-round-upper}

\input{r-round-lower}

\subsection*{Acknowledgements}
We thank the anonymous reviewers of SODA 2019 for helpful comments.

\bibliographystyle{abbrv}
\bibliography{general}

\end{document}

%% file: macros.tex
\newcommand{\Sec}[1]{Section~\ref{#1}}

\newcommand{\Eqn}[1]{Eq.~(\ref{#1})}

\newcommand{\Lem}[1]{Lemma~\ref{#1}}
\newcommand{\Thm}[1]{Theorem~\ref{#1}}

\newcommand{\Clm}[1]{Claim~\ref{#1}}

\newcommand{\cost}[1]{\ensuremath{\textnormal{\textsf{cost}}(#1)}}
\renewcommand{\c}{\ensuremath{\textnormal{\textsf{cost}}}}

\renewcommand{\c}{\ensuremath{\textnormal{\textsf{cost}}}}
\newcommand{\tS}{\ensuremath{\widetilde{\rS}}}

\newcommand{\bQ}{\ensuremath{\widetilde{Q}}}

\newcommand{\R}{{\mathsf R}}

\newcommand{\rT}{{\mathsf T}}

\newcommand{\E}{{\mathsf E}}
\newcommand{\A}{{\mathsf A}}
\newcommand{\X}{{\mathsf X}}

\renewcommand{\S}{{\mathsf S}}

\newcommand{\KAdapt}{\ensuremath{\textsc{$r$-Round-Adaptive}}}

\newcommand{\Reduce}{\ensuremath{\textsc{Reduce}}}
\newcommand{\Select}{\ensuremath{\textsc{Select}}}
\newcommand{\Round}{\ensuremath{\textsc{Round}}}

\newcommand{\1}{{\mathbf 1}}

\renewcommand{\>}{{\rightarrow}}

\renewcommand{\SS}{\ensuremath{\mathcal{S}}}

\newcommand{\nagreedy}{\ensuremath{\textnormal{\sc{non-adapt-greedy}}}\xspace}
\newcommand{\estar}{\ensuremath{e^{*}}}

\newcommand{\toShrink}{-.20cm}
\newcommand{\toShrinkEnu}{-.2cm}

\newcommand{\Leq}[1]{\ensuremath{\underset{\textnormal{#1}}\leq}}
\newcommand{\Le}[1]{\ensuremath{\underset{\textnormal{#1}}<}}
\newcommand{\Geq}[1]{\ensuremath{\underset{\textnormal{#1}}\geq}}
\newcommand{\Eq}[1]{\ensuremath{\underset{\textnormal{#1}}=}}

\renewcommand{\Xi}{\Psi}

\newcommand{\TT}{\ensuremath{\mathcal{T}}}

\newcommand{\rv}[1]{\ensuremath{\textnormal{\textsf{#1}}}}

\newcommand{\rX}{\rv{X}}
\newcommand{\rS}{\rv{S}}

\newcommand{\OPT}{\ensuremath{\textnormal{\textsf{OPT}}}\xspace}

\newcommand{\Ot}{\ensuremath{\widetilde{O}}}

\newcommand{\Bracket}[1]{\Big[#1\Big]}
\newcommand{\bracket}[1]{\left[#1\right]}
\newcommand{\paren}[1]{\ensuremath{\left(#1\right)}\xspace}
\newcommand{\card}[1]{\left\vert{#1}\right\vert}

\newcommand{\IN}{\ensuremath{\mathbb{N}}}

\newcommand{\set}[1]{\ensuremath{\left\{ #1 \right\}}}
\newcommand{\poly}{\mbox{\rm poly}}

\newcommand{\alg}{\ensuremath{\mathcal{A}}\xspace}

\DeclareMathOperator*{\Exp}{\ensuremath{{\mathbb{E}}}}
\DeclareMathOperator*{\Prob}{\ensuremath{\textnormal{Pr}}}
\renewcommand{\Pr}{\Prob}

\newcommand{\Ex}{\Exp}

\newenvironment{tbox}{\begin{tcolorbox}[
		enlarge top by=5pt,
		enlarge bottom by=5pt,
		 breakable,
		 boxsep=0pt,
                  left=4pt,
                  right=4pt,
                  top=10pt,
                  arc=0pt,
                  boxrule=1pt,toprule=1pt,
                  colback=white
                  ]
	}
{\end{tcolorbox}}

\newcommand{\event}{\mathcal{E}}


%% file: abstract.tex
\begin{abstract}
In the submodular cover problem, we are given a non-negative monotone submodular function $f$ over a ground set $E$ of items, and the goal is to choose a smallest subset $S \subseteq E$ such that $f(S) = Q$ where $Q = f(E)$. In the stochastic version of the problem, we are given $m$ stochastic items which are different random variables that independently realize to some item in $E$, and the goal is to find 
a smallest set of stochastic items whose realization $R$ satisfies $f(R) = Q$. The problem captures as a special case the stochastic set cover problem and more generally, stochastic covering integer programs. 

\smallskip

A fully adaptive algorithm for stochastic submodular cover chooses an item to realize and based on its realization, decides which item to realize next. A non-adaptive algorithm on the other hand needs to choose a permutation of items beforehand and
realize them one by one in the order specified by this permutation until the function value reaches $Q$. The cost of the algorithm in both case is the number (or costs) of items realized by the algorithm. It is not difficult to show that even for the coverage function there exist instances where the expected cost of a fully adaptive algorithm and a non-adaptive algorithm are separated by $\Omega(Q)$. This strong separation, often referred to as the {\em adaptivity gap}, is in sharp contrast to the separations observed in the framework of stochastic packing problems where the performance gap for many natural problem is close to the poly-time approximability of the non-stochastic version of the problem. Motivated by this striking gap between the power of adaptive and non-adaptive algorithms, we consider the following question in this work: does one need full power of adaptivity to obtain a near-optimal solution to stochastic submodular cover? In particular, how does the performance guarantees change when an algorithm interpolates between these two extremes using a few rounds of adaptivity. 

\smallskip

Towards this end, we define an {\em $r$-round adaptive} algorithm to be an algorithm that chooses a permutation of all available items in each round $k \in [r]$, and a threshold $\tau_k$, and realizes items in the order specified by the permutation until the function value is at least $\tau_k$. The permutation for each round $k$ is chosen adaptively based on the realization in the previous rounds, but the ordering inside each round remains fixed regardless of the realizations seen inside the round. Our main result is that for any integer $r$, there exists a poly-time $r$-round adaptive algorithm for stochastic submodular cover whose expected cost is $\tilde{O}(Q^{{1}/{r}})$ times the expected cost of a fully adaptive algorithm. Prior to our work, such a result was not known even for the case of $r=1$ and when $f$ is the coverage function. On the other hand, we show that for any $r$, there exist instances of the stochastic submodular cover problem where no $r$-round adaptive algorithm can achieve better than $\Omega(Q^{{1}/{r}})$ approximation to the expected cost of a fully adaptive algorithm. Our lower bound result holds even for coverage function and for algorithms with unbounded computational power. Thus our work shows that logarithmic rounds of adaptivity are necessary and sufficient to obtain near-optimal solutions to the stochastic submodular cover problem, and even few rounds of adaptivity are sufficient to sharply reduce the adaptivity gap.

\end{abstract}

%% file: intro.tex
\section{Introduction}\label{sec:intro}

Submodular functions naturally arise in many applications domains including algorithmic game theory, machine learning, and social choice theory, and have been extensively studied in combinatorial optimization. Many computational problems can be modeled as the {\em submodular cover} problem where we are given a non-negative monotone submodular function $f$ over a ground set $E$, and the goal is to choose a smallest subset $S \subseteq E$ such that $f(S) = Q$ where $Q = f(E)$. A well-studied special case is the set cover problem where the function $f$ is the coverage function and the items correspond to subsets of an underlying universe. Even this special case is known to be NP-hard to approximate to a factor
better than $\Omega(\log Q)$~\cite{DinurS14,Feige98,LundY94,Moshkovitz15}, and on the other hand, the classic paper of Wolsey~\cite{Wolsey82} shows that the problem admits a poly-time $O(\log Q)$-approximation for any integer-valued monotone submodular function.

In this work we consider the {\em stochastic version} of the problem that naturally arises when there is uncertainty about items. For instance, in stochastic influence spread in networks, the set of nodes that can be influenced by 
any particular node is a random variable whose value depends on the realized state of the influencing node (e.g. being successfully activated). 
In sensor placement problems, each sensor can fail partially or entirely with certain probability and the coverage of a sensor depends on whether the sensor failed or not.  
In data acquisition for machine learning (ML) tasks, each data point is apriori a random variable that can take different values, and one may wish to build a dataset representing a {diverse} set of values.  
For example, if one wants to build a ML model for 
identifying a new disease from gene patterns, one would start by building a database of 
gene patterns associated to that disease. In this case, each person's gene pattern 
is a random variable that can realize to different values depending on the race, gender, etc.
For other examples, we refer the reader to \cite{Liu+08} (application in databases) and \cite{Anagnostopoulos+15} (application
in document retrieval).

In the {\em stochastic submodular cover} problem, we are given $m$ stochastic items which are different random variables that independently realize to an element of $E$, and the goal is to find a lowest cost set of stochastic items whose realization $R$ satisfies $f(R) = Q$. 
In network influence spread problems
each item corresponds to a node in the network, and its realization corresponds to the set of nodes it 
can influence. In sensor placement problems an item corresponds to a sensor and its realization corresponds 
to the area that it covers upon being deployed.
In the case of data acquisition, an item corresponds to a data point and its realization
corresponds to the value it takes upon being queried. 
The problem captures as a special case the stochastic set cover problem and more generally, stochastic covering integer programs.

In stochastic optimization, a powerful computational resource is {\em adaptivity}. An {\em adaptive} algorithm for stochastic submodular cover chooses an item to realize and based on its realization, decides which item to realize next. A {\em non-adaptive} algorithm on the other hand needs to choose a permutation of items and realize them in the order specified by the permutation until the function value reaches $Q$. The cost of the algorithm in both cases is the number (or costs) of items realized by the algorithm. It is well-understood that in general, adaptive algorithms perform better than non-adaptive algorithms in terms of cost of coverage. However, 
in practical applications a non-adaptive algorithm is better from the point of view of
practitioners as it eliminates the need of sequential decision making and instead requires them to 
make just one decision. This motivates the study of separation between the performance of adaptive and non-adaptive algorithms, known as the {\em adaptivity gap}. For many stochastic packing problems, the adaptivity gap is only a constant. For instance, the adaptivity gap for budgeted stochastic max coverage where you are given a  constraint on the number of items that can be chosen and the goal is to maximize coverage, the adaptivity gap is bounded by $1 - 1/e$~\cite{Asadpour+08}. In a sharp contrast, for the covering version of the problem, it is not difficult to show an adaptivity gap of $\Omega(Q)$~\cite{GoemansVo06}.

Motivated by this striking separation between the power of adaptive and non-adaptive algorithms, we consider the following question in this work: does one need full power of adaptivity to obtain a near-optimal solution to stochastic submodular cover? In particular, how does the performance guarantees change when an algorithm interpolates between these two extremes using a few rounds of adaptivity.

Towards this end, we define an {\em $r$-round adaptive} algorithm to be an algorithm that chooses a permutation of all available items in each round $k \in [r]$, and a threshold $\tau_k$, and realizes items in the order specified by the permutation until the function value is at least $\tau_k$. 
A non-adaptive algorithm would then correspond to the case $r=1$ (with $\tau_1 = Q$), and 
an adaptive algorithm would correspond to the case $r=m$ (with $\tau_k = 0$ for all $k \in [r]$).
The permutation for each round $k$ is chosen adaptively based on the realization in the previous rounds, but the ordering inside each round remains fixed regardless of the realizations seen inside the round. 
We will call this the ``permutation framework'' for an $r$-round algorithm.

Our main result is that for any integer $r$, there exists a poly-time $r$-round adaptive algorithm for stochastic submodular cover whose expected cost is ${\Ot}(Q^{{1}/{r}})$ times the expected cost of a fully adaptive algorithm, where the ${\Ot}$ notation is hiding a logarithmic dependence on the number of items and the maximum cost of any item. Prior to our work, such a result was not known even for the case of $r=1$ and when $f$ is the coverage function. Indeed achieving such a result was cast as an open problem by Goemans and Vondrak~\cite{GoemansVo06} who achieved an $O(n^2)$ bound (corresponding to $O(Q^2)$) on the adaptivity gap of stochastic set cover. Furthermore, we show that for any $r$, there exist instances of the stochastic submodular cover problem where no $r$-round adaptive algorithm can achieve better than $\Omega(Q^{{1}/{r}})$ approximation to the expected cost of a fully adaptive algorithm. Our lower bound result holds even for coverage function and for algorithms with unbounded computational power. Thus our work shows that logarithmic rounds of adaptivity are necessary and sufficient to obtain near-optimal solutions to the stochastic submodular cover problem, and even few rounds of adaptivity are sufficient to sharply reduce the adaptivity gap.

\vspace{-1mm}

\begin{remark} 
\emph{One may consider an alternate notion of $r$-round adaptive algorithm: In each round $k \in [r]$, the algorithm chooses a fixed set of items to realize in parallel where the choice of the set depends on the realizations in the previous rounds (instead of 
a permutation over items). Let us call this framework the ``set framework''. One benefit of this variation is that items in each round can be realized in parallel. Unfortunately in this framework, any algorithm that always outputs a valid cover (as is our requirement), must in general include all remaining items in the last 
round, because for any proper subset of the remaining items there will be positive probability 
that this subset will not able to cover the entire set.
Hence, the $r$-round adaptivity gap would be $\Omega(m)$. }

\emph{Hence, one would have to consider a relaxed version of the problem and require that the algorithm achieves the desired coverage guarantee only with probability $1 - o(1)$. Our algorithmic results directly carry over to this variant of the problem. In particular, for any fixed $r$, we obtain poly-time $r$-round adaptive algorithm in the set framework whose cost is $\tilde{O}(Q^{{1}/{r}})$ times the expected cost of a fully adaptive algorithm, and that succeeds with probability at least $1 - o(1)$. At the same time, our lower
bound of $\Omega(Q^{{1}/{r}})$ continues to hold in this relaxed setting.
In the following we will provide results for only the permutation framework, with the understanding 
that all our results carry over to the set framework with the relaxed version of the problem.}
\end{remark}

\vspace{-1mm}

\subsection{Problem Statement}\label{sec:problem}

Let $\rX:=\set{\rX_1,\ldots,\rX_m}$ be a collection of $m$ \emph{independent} random variables each supported on the same ground set $E$ and $f$
be an integer-valued\footnote{We present our results for integer-valued functions for simplicity of exposition. All our results can easily be generalized to positive real-valued functions.} non-negative monotone submodular function $f: 2^{E} \rightarrow \IN_+$. 
We will refer to random variables $\rX_i$'s as items
and any set $\S \subseteq \rX$ as a set of items.
For any $i \in [m]$, we use $x_i \in E$ to refer to a realization of item (random variable) $\rX_i$ and define $X:=\set{x_1,\ldots,x_m}$ as the realization of $\rX$. 
We slightly abuse notation\footnote{Note that here $f: 2^{E} \rightarrow \IN_+$ is being extended to a function $f' : 2^{\rX} \rightarrow \IN_+$, but we chose to refer to $f'$ as $f$.} and extend $f$ to the ground set of items $\rX$ such that for any set $\S \subseteq \rX$, $f(\S) := f(\cup_{\rX_i \in \S} \rX_i)$: this definition means that for any
realization $S$ of $\S$, $f(S) = f(\cup_{x_i \in S} x_i)$. Finally, there is an integer-valued cost $c_i \in [C]$ associated with item $\rX_i \in \rX$. 



Let $Q := f(E)$. For any set of items $\rS \subseteq \rX$, we say that a realization $S$ of $\rS$ is \emph{feasible} iff $f(S) = Q$. 
We will assume that any realization $X$ of $\X$ is always feasible, i.e.\ $f(X) = Q$\footnote{One can ensure this by adding an item $\X_i$ to the ground set such that $f(x_i) = Q$ for all realizations $x_i$ of $\X_i$, but cost of this item is higher than the combined cost of all other items.}. We will say that a realization $X$ of $\X$ is \emph{covered} by a realization $S \subseteq X$ of $\S$
iff $S$ is feasible.
The goal in the stochastic submodular cover problem is to find a set of items $\rS \subseteq \rX$ with the minimum cost which gets realized to a feasible set.
In order to do so, if we include any item $\rX_i$ to $\rS$ we pay a cost $c_i$, and once included, $\rX_i$ would be realized to some $x_i \in E$ and is fixed from now on. Once a decision made regarding inclusion of an item in $\rS$, this item cannot be removed from $\rS$.

For any set of items $\S \subseteq \rX$, we define $\cost{\S}$ to be the total 
cost of all items in $\S$, i.e.\ $\cost{\S}= \sum_{i \in [m]} c_i \cdot \1[\rX_i \in \S]$, where $\1[\cdot]$ is 
an indicator function.
For any algorithm $\alg$, we refer to the total cost of solution $\S$ returned by $\alg$ on an instantiation $X$ of $\rX$ as the \emph{cost} of $\alg$ on $X$ denoted by $\cost{\alg(X)}$. We are interested in \emph{minimizing} the \emph{expected} cost
of the algorithm $\alg$, i.e., $\Ex_{X \sim \rX}\bracket{\cost{\alg{(X)}}}$. 

\begin{example}[\textbf{Stochastic Set Cover}]\label{exp:set-cover}
\vspace{1mm}
A canonical example of the stochastic submodular cover problem is the stochastic set cover problem.
Let $U$ be a universe of $n$ ``elements'' (not to be mistaken with ``items'') and $\rX = \set{\rX_1,\ldots,\rX_m}$ be a collection of $m$ random variables where each random variable $\rX_i$ is supported on subsets of $U$, i.e., 
realizes to some subset $T_i \subseteq U$. We refer to each random variable $\rX_i$ as a \emph{stochastic set}. In the stochastic set cover problem, the goal is to pick a smallest (or minimum weight) collection $\rS$ of items (or equivalently sets) in $\rX$ such that the realized sets in this collections cover the universe $U$.
\end{example}

We consider the following types of algorithms (sometimes referred to as policies in the literature) for the stochastic submodular cover problem: 
\begin{itemize}[itemsep=0pt,leftmargin=10pt]

\item \textbf{Non-adaptive:} 
A non-adaptive algorithm simply picks a fixed ordering of items in $\rX$ and insert the items one by one to $\rS$ until the realization $S$ of $\rS$ become feasible.


\item \textbf{Adaptive:} An adaptive algorithm on the other hand picks the next item to be included in $\rS$ adaptively based on the realization of previously chosen items. 
In other words, the choice of each item to be included in $\rS$ is now a function of the realization of items already in $\rS$.



\item \textbf{$\bm{r}$-round adaptive:} We define $r$-round adaptive algorithms as an ``interpolation'' between 
the above two extremes. For any integer $r \geq 1$, an $r$-round adaptive algorithm chooses the items to be included in $\rS$ in $r$ rounds of adaptivity: In each round $i \in [r]$, the algorithm 
chooses a \emph{threshold} $\tau_i \in \IN_+$ and an ordering over items, and then inserts the items one by one according to this ordering to $\rS$ until 
for the realized set $S$, $f(S) \geq \tau_i$. Once this round finishes, the algorithm decides on an ordering over
the remaining items \emph{adaptively} based on the current realization.


\end{itemize}

\noindent
In above definitions, a non-adaptive algorithm corresponds to case of $r=1$ round adaptive algorithm (with $\tau_1 = Q$) and 
a (fully) adaptive algorithm corresponds to the case of $r = m$ (here $\tau_i$ is irrelevant and can be thought as being zero). 

\paragraph{Adaptivity gap.} We use $\OPT$ to refer to the optimal adaptive algorithm for the stochastic submodular cover problem, i.e., an adaptive algorithm with minimum expected cost. We use the 
expected cost of \OPT as the main benchmark against which we compare the cost of other algorithms. In particular, we define \emph{adaptivity gap} as the ratio between the expected cost of the best non-adaptive algorithm
for the submodular cover problem and the expected cost of \OPT. Similarly, for any integer $r$, we define the \emph{$r$-round adaptivity gap} for $r$-rounds adaptive algorithms in analogy with above definition. 

\begin{remark}\label{rem:information-theoretic}
\emph{The notion of ``best'' non-adaptive or $r$-round adaptive algorithm defined above allow \emph{unbounded computational} power to the algorithm. Hence, the only limiting factor of the algorithm is the \emph{information-theoretic} 
barrier caused by the \emph{uncertainty} about the underlying realization.} 
\end{remark}

\subsection{Our Contributions}\label{sec:results}

In this paper, we establish \emph{tight} bounds (up to logarithmic factor) on the $r$-round adaptivity gap of the stochastic submodular cover problem for any integer $r \geq 1$. Our main result 
is an $r$-round adaptive algorithm (for any integer $r \geq 1$) for the stochastic submodular cover problem. 

\begin{result}[\textbf{Main Result}]\label{res:r-round-upper}
	For any integer $r \geq 1$ and any monotone submodular function $f$, there exists an $r$-round adaptive algorithm for the stochastic submodular cover problem for function $f$
	and set of items $\X$ with cost of each item bounded by $C$ that incurs expected cost $O(Q^{1/r} \cdot \log{Q} \cdot \log(mC))$ times the expected cost of the optimal adaptive algorithm. 
\end{result}
\noindent
A corollary of Result~\ref{res:r-round-upper} is that the $r$-round adaptivity gap of the submodular cover problem is $\Ot(Q^{1/r})$. This implies that using only $O\paren{\frac{\log{Q}}{\log\log{Q}}}$ rounds of 
adaptivity, one can reduce the cost of the algorithm to within \emph{poly-logarithmic} factor of the optimal adaptive algorithm. In other words, one can ``harness'' the (essentially) full power of adaptivity, in 
only logarithmic number of rounds.

Various stochastic covering problems can be cast as
submodular cover problem, including the stochastic set cover problem and the stochastic covering integer programs studied previously in the literature~\cite{GoemansVo06, GolovinKr10, Deshpande+14}. 
As such, Result~\ref{res:r-round-upper} directly extends to these problems as well. In particular, as a (very) special case of Result~\ref{res:r-round-upper}, we obtain that the 
adaptivity gap of the stochastic set cover problem is $\Ot(n)$ (here $n$ is the size of the universe), improving upon the $O(n^2)$ bound of Goemans and Vondrak~\cite{GoemansVo06} and settling
an open question in their work regarding the adaptivity gap of this problem (an $\Omega(n)$ lower bound was already shown in~\cite{GoemansVo06}).

We further prove that the $r$-round adaptivity gaps in Result~\ref{res:r-round-upper} are almost \emph{tight} for any $r \geq 1$. 

\begin{result}\label{res:r-round-lower}
	For any integer $r \geq 1$, there exists a monotone submodular function $f:2^{E} \rightarrow \IN_+$, in particular a coverage function, with $Q:= f(E)$ such that  
	the expected cost of any $r$-round adaptive algorithm for the submodular cover problem for function $f$, i.e., the stochastic set cover problem, is $\Omega(\frac{1}{r^3} \cdot Q^{1/r})$ times the expected cost of the optimal adaptive algorithm. 
\end{result}
Result~\ref{res:r-round-lower} implies that the $r$-round adaptivity gap of the submodular cover problem is $\Omega(\frac{1}{r^3} \cdot Q^{1/r})$, i.e., within poly-logarithmic factor of the upper bound in Result~\ref{res:r-round-upper}. 
An immediate corollary of this result is that $\Omega(\frac{\log{Q}}{\log\log{Q}})$ rounds of adaptivity are necessary for reducing the cost of the algorithms to within logarithmic factors of the optimal adaptive algorithm. We further point out that
interestingly, the optimal adaptive algorithm in instances created in Result~\ref{res:r-round-lower} only requires $r+1$ rounds; as such, Result~\ref{res:r-round-lower} in fact is proving a lower bound on the gap between
the cost of $r$-round and $(r+1)$-round adaptive algorithms. 

We remark that our algorithm in Result~\ref{res:r-round-upper} is \emph{polynomial time} (for polynomially-bounded item costs), while the lower bound in Result~\ref{res:r-round-lower} holds 
again algorithms with unbounded computational power (see Remark~\ref{rem:information-theoretic}).

\subsection{Related Work}\label{sec:related}
The problem of submodular cover was perhaps first studied by \cite{Wolsey82}, who showed that 
a \emph{greedy} algorithm achieves an approximation ratio of $\log(Q)$.
Subsequent to this there has been a lot of work on this problem 
in various settings \cite{GolovinKr10, AzarGa11, Azar+09, Im+16, Deshpande+14, Grammel+16, Kambadur+17}.
To our knowledge, the question of adaptivity in stochastic covering problems was
first studied in \cite{GoemansVo06} for the special case of 
stochastic set cover and covering integer programs.
It was shown that the adaptivity gap of this problem is $\Omega(n)$,
where $n$ is the size of the universe to be covered.
A non-adaptive algorithm for this problem
with an adaptivity gap of $O(n^2)$
was also presented.

Subsequently there has been a lot of work on stochastic set cover
and the more general stochastic submodular cover problem in the fully adaptive setting.
A special case of stochastic set cover was studied by \cite{Liu+08} in the adaptive setting,
and an \emph{adaptive greedy algorithm} was studied\footnote{The paper originally claimed
an approximation ratio of $\log(n)$ for this algorithm, however, the claim was later retracted by the authors 
due to an error in the original analysis \cite{Parthasarathy18}}.
In \cite{GolovinKr10} the notion of ``adaptive submodularity" was defined
for adaptive optimization,
which demands that given any partial realization of items, the marginal function 
with respect to this realization remains monotone submodular.
This paper also presented an \emph{adaptive greedy algorithm} for the problem 
of stochastic submodular cover, and stochastic submodular maximization subject to cardinality constraints.\footnote{It was originally claimed that this algorithm achieves an approximation ratio of $\log(Q)$ where $Q$ is 
the desired coverage, however, the claim was later retracted due to an error in the analysis \cite{NanSa17}. 
The authors have claimed an approximation ratio of $\log^2(Q)$ since then.} 
In \cite{Im+16} a more general version of stochastic submodular cover problem
was studied in the fully adaptive setting,
and their results imply the best-possible approximation ratio of $\log(Q)$ 
for stochastic submodular cover.
In \cite{Deshpande+14} an \emph{adaptive dual greedy} algorithm was presented
for this problem.  It was also shown that 
the \emph{adaptive greedy algorithm} of \cite{GolovinKr10} achieves an approximation ratio of $k \log (P)$, where $P$ is 
the maximum function value any item can contribute, and $k$ is the maximum support size 
of the distribution of any item.
There has also been work on this problem
when the realization of items can be correlated, unlike our setting where the 
realization of each item is independent.
In this setting, \cite{Kambadur+17} gives an adaptive algorithm 
which achieves an approximation ratio of $\log(Qs)$, where $Q$ is the desired 
coverage, and
$s$ denote the support size of the joint distribution 
of these correlated items.
In the case of independent realizations this quantity
will typically be exponential in the number of items.
In \cite{Grammel+16} a similar result was shown for a slightly different algorithm.

The question of adaptivity has also been studied for a related problem of \emph{stochastic submodular maximization}
 subject to cardinality constraints \cite{Asadpour+08}. 
The goal in this problem is to find a set of items with cardinality at most $k$,
 so as to maximize the expected value of a stochastic submodular function.
This paper showed that a 
non-adaptive greedy algorithm for this problem achieves an approximation ratio of $(1-\frac{1}{e})^2$
with respect to an optimal adaptive algorithm.
This result was later generalized to stochastic submodular maximization subject to matroid constraints
\cite{AsadpourNa16}.
In \cite{Gupta+17}, the adaptivity gap of
stochastic submodular maximization 
subject to a variety of \emph{prefix-closed constraints} was studied
under the setting where
the distribution of each item is Bernoulli.
This class of prefix-closed constraints includes matroid and knapsack constraints among others.
It was shown that there is a non-adaptive algorithm that
achieves an approximation ratio of $1/3$
with respect to an optimal adaptive algorithm.
In \cite{Hellerstein+15}, the problem of
stochastic submodular maximization 
was also studied under various types of constraints,
including knapsack constraints. 
An approximation ratio of $\tau$ for this problem under
 knapsack constraint was given, where $\tau$
is the smallest probability of any element in the ground set being 
realized by any item.  
The question of adaptivity has also been studied for other 
stochastic problems such as 
stochastic packing, knapsack, matching etc. (see, e.g.~\cite{Dean+05,Dean+08, YamaguchiMa18, Blum+15, Assadi+17, Assadi+16} and references therein).



There has also been a lot of work under the framework of $2$-stage or multi-stage stochastic programming \cite{Shapiro+09, SwamySh12, Charikar+05, ShmoysSw04}.
In this framework, one has to make sequential decisions in a stochastic environment, and there is a parameter 
$\lambda$, such that the cost of making the same decision increases by a factor $\lambda$
after each stage.
The stochastic program in each stage is defined 
in terms of the expected cost in the later stages.
The central question in these problems is-- when can we find good solutions to this
complex stochastic program, either by directly solving it or by 
finding approximations to it?
This largely depends on the complexity of the stochastic program at hand.
For example, if the distribution of the environment is explicitly given, then one might be 
able to solve the stochastic program exactly by using 
integer programming, and this question becomes largely \emph{computational} in nature.
This is fundamentally different than the \emph{information theoretic} question we consider in this paper.

Aside from the stochastic setting, algorithms with limited adaptivity have been studied across a wide spectrum of areas in computer science including in sorting and selection (e.g.~\cite{Valiant75,Cole86,BravermanMW16}), multi-armed bandits (e.g.~\cite{PerchetRCS15,Agarwal+17}), algorithms design (e.g.~\cite{BalkanskiSi18,Emamjomeh16,EneN18,BalkanskiRS18}), among others; we refer the interested reader to these papers and references therein for more details.

\begin{remark} 
\emph{Our study of $r$-round adaptive algorithm for submodular cover is reminiscent of a recent work of Chakrabarti and Wirth~\cite{ChakrabartiW16} on multi-pass streaming algorithms 
for the set cover problem. They showed that allowing additional passes over the input in the streaming setting (similar-in-spirit to more rounds of adaptivity) 
can significantly improve the performance of the algorithms and established tight pass-approximation tradeoffs that are similar (but not identical) to $r$-round adaptivity gap bounds
in Results~\ref{res:r-round-upper} and Results~\ref{res:r-round-lower}. In terms of techniques, our upper bound result---our main contribution---is almost entirely disjoint from the techniques in~\cite{ChakrabartiW16} (and works for the more general problem of submodular cover, whereas the results in~\cite{ChakrabartiW16} are specific to set cover), while
our lower bound uses similar instances as~\cite{ChakrabartiW16} but is based on an entirely different analysis. }
\end{remark}

\subsection{Organization}
In \Sec{sec:prelim} we present some preliminaries for our problem. In \Sec{sec:technical} we present 
a technical overview of our main results. In \Sec{app:select} we present a non-adaptive selection algorithm 
that will be used to prove our upper bound result in \Sec{sec:upper}.
We present the lower bound result in \Sec{sec:lower}.

%% file: prelim.tex
\section{Preliminaries}\label{sec:prelim}

\paragraph{Notation.} 
Throughout this paper we will use symbols $S, T,$ and $R$ to denote subsets of the ground set $E$, and use symbols $A$ and $B$ to denote subsets of $[m]$, i.e., indices of items. We will also use symbols $\S, \rT$ and $\R$ to denote subsets of $\X$ which realize to subsets $S, T$ and $R$ of the ground set $E$.

\medskip
\noindent
{\bf Submodular Functions:}\label{sec:submodular}
Let $E$ be a finite ground set and $\IN_+$ be the set of non-negative integers. For any set function $f: 2^{E} \rightarrow \IN_+$, and any set $S \subseteq E$, 
we define the \emph{marginal contribution} to $f$ as the set function $f_S: 2^{E} \rightarrow \IN_+$ such that for all $T \subseteq E$, 
\begin{align*}
	f_S(T) = f(S \cup T) - f(S). 
\end{align*}

When clear from the context, we abuse the notation and for $e \in E$, we use $f(e)$ and $f_S(e)$ instead of $f(\set{e})$ and $f_S(\set{e})$. 

A set function $f :2^{E} \rightarrow \IN_+$ is submodular iff for all $S \subseteq T \subseteq E$ and $e \in E$: $f_S(e) \geq f_T(e)$. 
Function $f$ is additionally monotone iff $f(S) \leq f(T)$. Throughout the paper, we solely focus on monotone submodular functions 
unless stated explicitly otherwise. 

We use the following two well-known facts about submodular functions throughout the paper. 

\begin{fact}\label{fact:marginal-greedy}
	Let $f(\cdot)$ be a monotone submodular function, then:
	\begin{align*}
		\forall S,T \subseteq E \qquad f(S) \leq f(T) + \sum_{e \in S \setminus T} f_T(e). 
	\end{align*} 
\end{fact}
\begin{fact}\label{fact:marginal-submodular}
	Let $f(\cdot)$ be a monotone submodular function, then for any $S \subseteq E$, $f_S(\cdot)$ is also monotone submodular. 
\end{fact}

%% file: technical.tex
\section{Technical Overview}\label{sec:technical}

We give here an overview of the techniques used in our upper and lower bound results.

\subsection{Upper Bound on $r$-round Adaptivity Gap} 
In this discussion we focus mainly on our non-adaptive ($r=1$) algorithm, which already deviates significantly 
from the previous work of Goemans and Vondrak~\cite{GoemansVo06}. A non-adaptive algorithm simply picks a permutation of items and realize them one by one in a set $S$ until $f(S) = Q$.
Hence, the ``only'' task in designing a non-adaptive algorithm is to find a ``good'' ordering of items, that is, an ordering such that its prefix that covers $Q$ has a low expected cost.

Consider the following algorithmic task: In the setting of stochastic submodular cover problem, suppose we are given
a (ordered) set $\rS$ of stochastic items. Can we pick a low-cost (ordered) set $\rT$ of stochastic items non-adaptively (without looking at a realization of $\rS$ or $\rT$) 
so that the coverage of $\rS \cup \rT$ is sufficiently larger than $\rS$, i.e., $\Ex\bracket{f_{\rS}(\rT)}$ is large? Assuming we can do this, we can use this primitive to find sets with large coverage non-adaptively and iteratively, by starting from the empty-set and using this primitive to increase the coverage further repeatedly.

Recall that in the non-stochastic setting, the greedy algorithm is precisely solving this problem,
i.e., finds a set $T$ such that $\frac{f_{S}(T)}{\cost{T}} \geq \frac{Q - f(S)}{\cost{\OPT}}$, where with a slight abuse of notation, $\OPT$ here denotes the optimal non-stochastic cover of $f(E)$. 
This suggests that one can always find a ``low'' cost set $T$ with a large marginal contribution to $S$. 
For the stochastic problem, however, it is not at all clear whether there always 
exists a ``low cost'' (compared to adaptive $\OPT$) $\rT$ whose expected marginal contribution to $\S$ is large. This is because there are many different realizations possible for 
$\rS$, and each realization $S$, in principle may require a \emph{dedicated} set of items $\rT(S)$ to achieve a large value $\Ex\bracket{f_S(\rT(S)) \mid S}$. As such, while adaptive \OPT can 
first discover the realization $S$ of $\rS$ and based on that choose $\rT(S)$ to increase the expected coverage, a non-adaptive algorithm needs to instead pick $\cup_{S \in \rS} \rT(S)$, which can have a much larger cost (but the same marginal contribution). 
This suggests that cost of non-adaptive algorithm can potentially grow with the size of all possible realizations of $S$. We point out that this task remains challenging 
even if all remaining inputs other than $\rS$ are non-stochastic, i.e., always realize to a particular item.  

Nevertheless, it turns out that no matter the size of the set of all realizations of $\rS$, one can always find a set of stochastic items $\rT$ such that 
$\Ex\bracket{f_{\rS}(\rT)} = \Omega(1) \cdot \Ex\bracket{Q-f(\rS)}$ while $\cost{\rT} = \Ot(Q) \cdot \Ex\bracket{\cost{\OPT}}$, i.e., achieve a marginal contribution proportional to $\Ex\bracket{Q-f(\rS)}$ while paying cost
which is $\Ot(Q)$ times larger than $\OPT$ (here $\OPT$ corresponds to an optimal adaptive algorithm corresponding the residual problem of covering $Q-f(\rS)$). Compared to the non-stochastic setting,
this cost is $\Ot(Q)$ times larger than the analogous cost in the non-stochastic setting (see Example~\ref{exp:deficit}). This part is one of the main technical ingredients of our paper (see Theorem~\ref{thm:select}). We briefly describe the main ideas behind this proof. 

The idea behind our algorithm is to \emph{sample} several realizations $S_1,\ldots,S_{\Xi}$ from $\rS$ and pick
a low cost dedicated set $\rT_i$ for each $S_i$ such that $\Ex\bracket{f_{S_i}(\rT_i)}$ is large (here, the randomness is only on realizations of $\rT_i$). This step is quite similar to
solving the non-adaptive submodular maximization problem with knapsack constraint for which we design a new algorithm based on an adaptation of Wolsey's LP~\cite{Wolsey82} (see Theorem~\ref{thm:detsel}
and discussion before that for more details and comparison with existing results). This allows us to bound the cost of each set $\rT_i$ by $O(\Ex\bracket{\cost{\OPT}})$. The final (ordered) set returned
by this algorithm is then $\rT := \rT_1 \cup \ldots \cup \rT_\Xi$. 
The 
ordering within items of $\rT$ does not matter.

The \emph{main} step of this argument 
is however to bound the value of $\Xi$, i.e., the number of samples, by $O(Q)$. This step is done by bounding the total 
contribution of sets $\rT_1,\ldots,\rT_{\Xi}$ on their own, i.e., $\Ex\bracket{f(\rT_1 \cup \ldots \cup \rT_\Xi)}$ independent of the set $\rS$. 
The intuition is that if we choose, say $\rT_1$, with respect to some realization $S$ of $\rS$, but $\rT_1$ does not have a marginal contribution 
to most realizations $S'$ of $\rS$, then this means that by picking another set $\rT_{2}$, the set $\rT_1 \cup \rT_2$ needs to have a coverage larger than both $\rT_1$ and $\rT_2$. 
As a result, if we repeat this process sufficiently many times, we should eventually be able to increase $\Ex\bracket{f_{\rS}(\rT)}$, simply because otherwise $f(\rT) > Q$, a contradiction. 

We now use this primitive to design our non-adaptive algorithm as follows: we keep adding set of items to the ordering using the primitive above in
\emph{iterative phases}. In each phase $p$, we run the above primitive multiple times to find a set $\rS_p$ with $\Ex\bracket{Q-f(\rS_p) \mid \event_{p-1}} = o(1)$, where $\event_{p-1}$ is the event that
the realization of items picked in previous phases of the algorithm did not cover $Q$ entirely. We further bound the cost of the set $\rS_p$ with the expected cost of $\OPT$ conditioned on the event $\event_{p-1}$, i.e.,
$\Ex\bracket{\cost{\OPT} \mid \event_{p-1}}$. Notice that this quantity can potentially be much larger than the expected cost of $\OPT$. However, since the probability that in the permutation returned by the non-adaptive algorithm, 
we ever need to realize the sets in $\rS_{p}$ is bounded by $\Pr\paren{\event_{p-1}}$, we can \emph{pay} for the cost of these sets in expectation. By repeating these phases, we can reduce the probability 
of not covering $Q$ \emph{exponentially fast} and finalize the proof. 

We then extend this algorithm to an $r$-round adaptive algorithm for any $r \geq 1$. 
For simplicity, let us only mention the extension to $2$ rounds (extending to $r$ is then straightforward). We spend the first round to find a (ordered) set $\rS$ with $f(S) \geq Q - \sqrt{Q}$ with high probability for any realizations $S$ of $\rS$. 
We extend our main primitive above to ensure that if $\Ex\bracket{Q-f(\rS)} \geq \sqrt{Q}$, then we can find a set $\rT$ with $\Ex\bracket{f_{\rS}(\rT)} = \Omega(1) \cdot \Ex\bracket{Q-f(\rS)}$ and
 $\cost{\rT} = \Ot(\sqrt{Q}) \cdot \Ex\bracket{\cost{\OPT}}$ (as opposed to $O(Q)$ in the original statement). This is achieved by the fact that when the deficit $Q - f(\rS)$ is sufficiently large 
then the rate of coverage per cost is higher, as 
opposed to when the deficit $Q - f(\rS)$ is very small.
Precisely, we exploit the fact that the gap of $Q - f(\rS)$ is sufficiently large 
to reach the contradiction in the original argument with only $O(\sqrt{Q})$ sets  $\rT_1,\rT_2,\ldots$. We then run the previous algorithm using this primitive by setting the threshold $\tau_1 = Q - \sqrt{Q}$. 
 In the next round, we simply run our previous algorithm on the function $f_{S}(\cdot)$ where $S$ is the realization in the first round.
As $f_S(\cdot)$ has maximum value at most $O(\sqrt{Q})$, by the previous argument we only need to pay $\Ot(\sqrt{Q})$ times expected cost of $\OPT$, hence our total cost is $\Ot(\sqrt{Q}) \cdot \Ex\bracket{\cost{\OPT}}$. 
Extending this approach to $r$-round algorithms is now straightforward using similar ideas as the thresholding greedy algorithm for set cover (see, e.g.~\cite{CormodeKW10}).

\subsection{Lower Bound on Adaptivity Gap}
 We prove our lower bound for the stochastic set cover problem, a special case of stochastic submodular cover problem (see Example~\ref{exp:set-cover}). 
Let us first sketch our lower bound for two
round algorithms. Let $\SS:=\set{U_1,\ldots,U_k}$ be a collection of $k = \poly{(n)}$ sets to be determined later (recall that $n$ is the size of the universe $U$ we aim to cover). 
Consider the following instance of stochastic set cover: there exists a single stochastic set $\rT$ which realizes to one set chosen uniformly at random
from sets $\overline{U_1},\ldots,\overline{U_k}$, i.e., complements of the sets in $\SS$. We further have $k$ additional 
stochastic sets where $\rT_i$ realizes to $U_i \setminus \set{e}$ for $e$ chosen uniformly at random from $U_i$.
Finally, for any element $e \in U$, we have a set $\rT_e$ with only one realization which is the singleton set $\set{e}$ (i.e., $\rT_e$ always covers $e$). 

Consider first the following adaptive strategy: pick $\rT$ in the first round and see its realization, say, $\overline{U_i}$. Pick $\rT_i$ in the second round and see its realization, say $U_i \setminus \set{e}$. Pick $\rT_e$ in the third round. 
This collection of sets is $\paren{U \setminus U_i} \cup \paren{U_i \setminus \set{e}} \cup (\set{e}) = U$, hence it is a feasible cover. As such, in only $3$ rounds of adaptivity, we were able to find a solution with cost only $3$. 

A two-round algorithm is however one round short of following the above strategy. One approach to remedy this would be try to make a ``shortcut''  by picking more than 
one sets in each round of this process, e.g., pick the set $\rT_i$ also in the first round. 
However, it is easy to see that as long as we do not pick $\Omega(k)$ sets in the first round, or $\Omega(\card{U_i})$ sets in the second round, we have a small chance of making such a shortcut. We are not done yet as it is possible
that the algorithm covers the universe using entirely different sets (i.e., do not follow this strategy). To ensure that cannot help either, we need the sets in $U_1,\ldots,U_k$ to have ``minimal'' intersection; this in turns
limits the size of each set $U_i$ and hence the eventual lower bound we obtain using this argument.

We design a family of instances that allows us to extend the above argument to $r$-round adaptive algorithms. 
We construct these instances using the \emph{edifice} set-system of Chakrabarti and Wirth~\cite{ChakrabartiW16}
that poses a ``near laminar'' property, i.e., any two sets are either subset-superset of one another or have ``minimal'' intersection. We remark that this set-system was originally introduced by~\cite{ChakrabartiW16} for designing 
multi-pass streaming lower bounds for the set cover problem. While the instances we create in this work are similar to the instances of~\cite{ChakrabartiW16}, the proof of our lower bound is entirely different (lower bound of
\cite{ChakrabartiW16} is proven using a reduction in communication complexity).

%% file: non-adaptive-selection.tex
\section{The Non-Adaptive Selection Algorithm}
\label{app:select}

We introduce a key primitive of our approach in this section for solving the following task: 
Suppose we have already chosen a subset $\rS \subseteq \rX$ of items but we are not aware of the realization of these
items; our goal is to non-adaptively add another set $\rT$ to $\rS$ to increase its expected coverage. 
Formally, given any monotone submodular function $g: 2^{E} \> \IN_+ $, let $Q_g := g(E)$ be the required coverage on $g$.
Also, for any realization $S$ of $\rS$, we use $\Delta(S) := Q_g - g(S)$ to refer to the \emph{deficit} in covering $Q_g$, and denote by $\Delta := \Ex\bracket{\Delta(\rS)}$ the expected deficit of set $\rS$. 
Our goal is now to add (still non-adaptively) a ``low-cost'' (compared to adaptive \OPT) set $\rT $ to $\rS$ to decrease the \emph{expected} deficit. It is easy to see that such a primitive would be helpful for finding 
sets with ``large'' coverage non-adaptively and iteratively, by starting from the empty-set and use this primitive to reduce the deficit further by picking another set and then repeat the process starting from this set. 

Let us start by giving an example which shows some of the difficulty of this task. 

\begin{example}\label{exp:deficit}
\vspace{1mm}
	Consider an instance of stochastic set cover: there exists a single set, say $\rX_1$ which realizes to $U \setminus \set{\estar}$ for an element $\estar$ chosen uniformly at random from $U$
	and $n$ singleton sets $\rX_2, \cdots \rX_{n+1}$, each covering a unique element in $U$. If we have already chosen $\rX_1$, and want to chose more sets in order to decrease the expected deficit, then it is easy to see that
	even though the cost of $\OPT$ is only $2$, no collection of $o(n)$ sets can decrease the expected deficit by one. This should be contrasted with the non-stochastic setting in which there always exists
	a \emph{single} set that reduces a deficit of $\Delta$ by $\Delta/\cost{\OPT}$. 
\end{example}

We are now ready to state our main result in this section. 
\begin{theorem}
\label{thm:select}
Let $\rX$ be a collection of items, and let $g$ be any monotone submodular function such that $ g(X) = Q_g$
for every realization $X$ of $\X$.
Let $\rS \subseteq \rX$ be any subset of items and define $\Delta:=\Ex\bracket{Q_g - g(\rS)}$. Given any parameter $\alpha \geq Q_g/\Delta$, there is a \emph{randomized non-adaptive} algorithm
that outputs a set $\rT \subseteq \rX \setminus \rS$ such that cost of $\rT$ is $O(\alpha) \cdot \Ex\bracket{\cost{\OPT}}$ in expectation over the randomness of the algorithm and 
$\Ex\bracket{Q_g - g(\rS \cup \rT)} \leq 5\Delta/6$ over the randomness of the algorithm and realizations of $\rS$ and $\rT$. 
Here $\OPT$ is an optimal fully-adaptive algorithm for the stochastic submodular cover problem with the function $g$\footnote{Throughout this paper we will abuse notation by refering to an optimal 
fully-adaptive algorithm for different problem instances using the same notation $\OPT$. The specific problem instance will be clear from context.}.
\end{theorem}

The goal in Theorem~\ref{thm:select}, is to select a set of items that can decrease the deficit of a \emph{typical} realization $S$ of $\rS$ (i.e., the expected deficit).
In order to do so, we first design a non-adaptive algorithm that finds a low-cost set that can decrease the deficit of a \emph{particular} realization $S$ of $\rS$. This step is closely related 
to solving a stochastic submodular maximization problem subject to a knapsack constraint. Indeed, when costs of all the items are the same, i.e., when we want to minimize the number of items in the solution, 
one can use the algorithm of~\cite{Asadpour+08} (with some small modification) for stochastic submodular maximization subject to cardinality constraint for this purpose. 
Also, when the random variables $\rX_i$'s have binary realizations,
i.e.\ take only two possible values, then one can use the algorithm of \cite{Gupta+17} for this purpose.
However, we are not aware of a solution for the knapsack constraint of the problem in its
general form
with the bounds required in our algorithms, and hence we present an algorithm for this task as well.  The \emph{main step} of our argument is however on how to use this algorithm to prove Theorem~\ref{thm:select}, i.e., 
move from per-realization guarantee, to the expectation guarantee. 

\subsection{A Non-Adaptive Algorithm for Increasing Expected Coverage}\label{sec:maximization} 

We start by presenting a non-adaptive algorithm that picks a low-cost (compared to the \emph{expected} cost of \OPT) set of items \emph{deterministically},
while achieving a constant factor of \emph{coverage} of \OPT. 
For any set $A \subseteq [m]$, i.e., the set of indices of stochastic items, and any realization $X$ of $\rX$, we define $X_A := \set{x_i \mid i \in A}$, i.e, the realization of all items corresponding to indices in $A$.

\begin{theorem}\label{thm:detsel}
There exists a non-adaptive algorithm that takes as input a
set of items $\X$, a monotone submodular function $f$, and a parameter $\bQ$ such that $f(X) = \bQ$ for any realization $X$ of $\X$, and outputs a set $A \subseteq [m]$
such that $(i)$ $\c(\rX_A) \leq 3\cdot\Ex\bracket{\c(\OPT)}$ 
and $(ii)$ $\Ex_{X_A \sim \rX}\bracket{f(X_A)} \geq {\bQ}/{3}$.
Here, $\OPT$ is the optimum adaptive algorithm for submodular cover on $\X$ with function $f$ and parameter $Q = \bQ$.
\end{theorem}

As argued before, Theorem~\ref{thm:detsel} can be interpreted as an algorithm for submodular maximization subject to {knapsack constraint}. 

To prove Theorem~\ref{thm:detsel}, we design a simple greedy algorithm (similar to the greedy algorithm for submodular maximization) and analyze it using a linear programming (LP) relaxation
in the spirit of Wolsey's LP \cite{Wolsey82} defined in the following section. 

\subsection*{Extension of Wolsey's LP for Stochastic Submodular Cover}
Let us define the function $F: 2^{[m]} \> \IN_+$ as follows: for any $A \subseteq [m]$,
\begin{align}
 F(A) := \Ex_{X_A \sim \rX}\bracket{f(X_A)}. \label{eq:F}
\end{align}
As we assume in the lemma statement that $\bQ := \Ex_{X \sim \X }[f(X)]$, we have $F([m]) = \bQ$ as well. 
For any $B \subseteq [m]$, we further define the marginal contribution function $F_B: 2^{[m]} \> \IN_+$ where $F_B(A) := F(A \cup B) - F(B)$ for all $A \subseteq [m] \setminus B$. 
The following proposition is straightforward. 
\begin{proposition}
\label{prop:exp-submodular}
Function $F$ is a monotone submodular function.
\end{proposition}
\begin{proof}
$F$ is a convex combination of submodular functions, one for each realization of $\X$.
\end{proof}

We will use a linear programming (LP) relaxation in the spirit of Wolsey's LP \cite{Wolsey82} for the submodular
cover problem (when applied to the function $F$). Consider the following linear programming relaxation:
\begin{tbox}
\vspace{-15pt}
\begin{align}
&P =   \min_{y \in [0,1]^m}  \sum_{i =1}^m c_i \cdot y_i	\notag \\
&\text{s.t.} \sum_{i \in [m] \setminus A}  F_{A}(i) \cdot y_i \geq \bQ  - 2F(A), \quad \forall A \subseteq [m] \label{eqn:lp}
\end{align}
\end{tbox}

The difference between LP~(\ref{eqn:lp}) and Wolsey's LP is in RHS of the constraint which is $\bQ - F(A)$ in case Wolsey's LP. In the non-stochastic setting, 
one can prove that Wolsey's LP lower bounds the value of optimum submodular cover for function $F$. To extend this result to the stochastic case (for the function $f$) however, it suffices to modify the constraint 
as in LP~(\ref{eqn:lp}), as we prove in the following lemma. 


\begin{lemma}
\label{lem:lower-bound-opt}
The cost of an optimal adaptive algorithm \OPT for submodular cover on function $f$ is lower bounded by the optimal cost $P$ of LP~(\ref{eqn:lp}), i.e. $P \leq \Ex\bracket{\c{(\OPT)}}$.
\end{lemma}
\begin{proof}
For a realization $X$ of $\X$ and any $i \in [m]$, define an indicator random variable $w_i(X)$
that takes value $1$ iff $\OPT$ chooses $\X_i$ on the realization $X$, i.e.
\[
w_i(X) = \1[ \X_i \in \OPT(X)].
\]
Let $w_i$ be the probability that $\OPT$ chooses $\X_i$, i.e.,
\[
w_i = \Pr_{X \sim \rX} \paren{w_i(X) = 1} = \Ex\bracket{w_i(X)}.
\]
We have that, 
\begin{align*}
	\Ex\bracket{\cost{\OPT}} &= \Ex_{X}\bracket{\sum_{i=1}^{m} \1[ \X_i \in \OPT(X)] \cdot c_i}  = \sum_{i=1}^{m} w_i \cdot c_i
			\,. 
\end{align*}
In the following, we prove that $w:=(w_1,\ldots,w_m)$ is a feasible solution to LP~(\ref{eqn:lp}), which by above equation would immediately imply that $P \leq \Ex\bracket{\cost{\OPT}}$. 

Clearly $w \in [0,1]^{m}$, so it suffices to prove that the constraint holds for any set $A \subseteq [m]$. The main step in doing so is the following claim. 
\begin{claim}\label{clm:constraint-holds}
	For any set $A \subseteq [m]$, and any two realizations $X$ and $X'$ of $\rX$: 
	\begin{align*}
		f(X_A) + f(X'_A) + \sum_{i \in [m] \setminus A} f_{X'_A}(x_i) \cdot w_i(X) \geq \bQ. 
	\end{align*}
\end{claim}
\begin{proof}
	Recall that we assume $f(X) = \bQ$ always, and hence $f(\OPT(X)) = \bQ$ as well. Moreover, for any $i \in \OPT(X)$, $w_i(X) = 1$ and for $i \in [m] \setminus \OPT(X)$, $w_i(X) = 0$.
	We further define the sets: 
	\begin{align*}
		B:= \OPT(X) \cap A \qquad \text{and} \qquad C:= \OPT(X) \setminus B. 
	\end{align*}
	We have, 
	\begin{align*}
		f(X_A) + f(X'_A) + \sum_{i \in [m] \setminus A} f_{X'_A}(x_i) \cdot w_i(X)  &= f(X_A) + f(X'_A) + \sum_{x_i \in C} f_{X'_A}(x_i) \\
		&\!\!\!\!\!\!\Geq{Fact~\ref{fact:marginal-greedy}} f(X_A) + f(X'_A \cup C) \tag{by submodularity} \\
		&\geq f(X_B) + f(X_C) \tag{by monotonicity as $X_B \subseteq X_A$} \\
		&= f(X_B \cup X_C) = \bQ \tag{by submodularity and since $X_B \cup X_C = \OPT(X)$},
	\end{align*}
	which finalizes the proof.
	\Qed{Claim~\ref{clm:constraint-holds}}
	
\end{proof}

Fix any set $A \subseteq [m]$. We first take an expectation over all realizations of $X$ in LHS of Claim~\ref{clm:constraint-holds}:
\begin{align*}
	\bQ &\Leq{Claim~\ref{clm:constraint-holds}} \Ex_{X}\Bracket{f(X_A) + f(X'_A) + \sum_{i \in [m] \setminus A} f_{X'_A}(x_i) \cdot w_i(X)}  \\
	&= \Ex_{X}\bracket{f(X_A)} + f(X'_A)  + \sum_{i \in [m] \setminus A} \Ex_{X}\bracket{f_{X'_A}(x_i) \cdot w_i(X)} \\
	&= \Ex_{X}\bracket{f(X_A)} + f(X'_A) 	 + \sum_{i \in [m] \setminus A} \Ex_{X}\bracket{f_{X'_A}(x_i)} \cdot \Ex_{X}\bracket{w_i(X)}, 
\end{align*}
as random variables $f_{X'_A}(\rX_i)$ and $w_i(\rX)$ are independent since the choice of $\rX_i$ by $\OPT$ is independent of what $\rX_i$ realizes to. We further point out that $\Ex_{X}\bracket{f(X_A)}$ in the RHS of 
last equation above is equal to $F(A)$ by definition in Eq~(\ref{eq:F}) and $\Ex_{X}\bracket{w_i(X)} = w_i$. 
 
We further take an expectation over all realizations of $X'$ in the RHS above:
\begin{align*}
	\bQ &\leq \Ex_{X'}\Bracket{F(A) + f(X'_A) + \sum_{i \in [m] \setminus A} \Ex_{X}\bracket{f_{X'_A}(x_i)} \cdot w_i} \\
	&\!\!\!\!\Eq{Eq~(\ref{eq:F})} F(A) + F(A) + \sum_{i \in [m] \setminus A} \Ex_{X'}\Ex_{X}\bracket{f_{X'_A}(x_i)} \cdot w_i \\
	&= 2 \cdot F(A) + \sum_{i \in [m] \setminus A} F_{A}(i) \cdot w_i 
		\,, 
\end{align*}
as $F_A(i) = \Ex_{X'}\Ex_{X}\bracket{f(X'_A \cup X_i) - f(X'_A)}$. Rewriting the above equation, we obtain that the constraint associated with set $A$ is satisfied by $w$. 
This concludes the proof that $w$ is a feasible solution. 
\Qed{Lemma~\ref{lem:lower-bound-opt}}

\end{proof}

\subsection*{The Non-Adaptive-Greedy Algorithm}
We now design an algorithm, namely $\nagreedy$, based on ``the greedy algorithm'' (for submodular optimization) applied to the function $F$ in the last section and then use LP~(\ref{eqn:lp}) to analyze it. We emphasize
that the use of the LP is only in the analysis and not in the algorithm.

\begin{tbox}
\label{algo:non-adaptive}
{$\nagreedy(\rX,f,\bQ)$. Given a monotone submodular function $f$, the set of stochastic items $\rX$, and a parameter $\bQ = f(X)$ for all $X$, outputs a set $A$ of (indices of) stochastic items.} 
\begin{enumerate}
\item \textbf{Initialize:} Set $A \leftarrow \emptyset$ and $F$ be the function associated to $f$ in Eq~(\ref{eq:F}).
\item \textbf{While} $F(A) < \bQ/3$ \textbf{do:}
\begin{enumerate}
\item Let $j^* \leftarrow \argmax_{j \in [m]} {F_{A}(j)}/{c_j}$. 
\item Update $A \leftarrow A \cup \set{j^*}$. 
\end{enumerate} 
\item \textbf{Output:} $A$.
\end{enumerate}
\end{tbox}


It is clear that the set $A$ output by $\nagreedy$ achieves $\Ex_{X_A}\bracket{f(X_A)} = F(A) \geq \bQ/3$ (as $F([m]) = \bQ$, the termination condition would always be satisfied eventually). 
We will now bound the cost paid by the greedy algorithm in terms of the optimal value $P$ of LP~(\ref{eqn:lp}). 

\begin{lemma}\label{lem:bound-A}
	$\cost{\rX_A} \leq 3 P$. 
\end{lemma}

To prove Lemma~\ref{lem:bound-A} we need some definition. 
Let the sequence of items picked by the greedy algorithm be $j_1 , j_2 , j_3 \cdots $,
where $j_i$ is the index of the item picked in iteration $i$. Moreover, for any $i$, define $A_{<i}:=\set{j_1,\ldots,j_{i-1}}$, i.e., the set of items chosen before iteration $i$.
We first prove the following bound on the ratio of coverage rate to costs in each iteration.

\begin{lemma}
\label{lem:greedy-progress}
In each iteration $i$ of the non-adaptive greedy algorithm we have, 
\[
\frac{F_{A_{<i}}(j_i)}{c_{j_i}} \geq \frac{\bQ - 2 F(A_{<i})}{P},
\]
where $P$ is the optimal value of LP~(\ref{eqn:lp}).
\end{lemma}
\begin{proof}
Fix any iteration $i$. Recall that in each iteration, we pick the item $j_i \in \argmax_{j \in [m]} {F_{\A_{<i}}(j)}/{c_j}$. 
Suppose towards a contradiction that in some iteration $i$:
\begin{align}
	\forall j \in [m] \qquad \frac{{F_{\A_{<i}}(j)}}{{c_j}} < \frac{\bQ - 2F(A_{<i})}{P}. \label{eqn:assume-greedy-progress}
\end{align}
\noindent
Let $y^*$ be an optimal solution to LP~(\ref{eqn:lp}), then by the constraint of the LP for set $A_{<i}$ we have
\begin{align*}
\bQ - 2 F(A_{<i}) &\leq \sum_{j \in [m] \setminus A_{<i}}  F_{A_{<i}}(j)	 \cdot  y^*_j\\
			&\!\!\!\! \Le{Eq~(\ref{eqn:assume-greedy-progress})} \sum_{j \in [m] \setminus A_{<i}} y^*_j \cdot c_j \cdot \frac{\bQ - 2 F(A_{<i})}{P} \\
			& \leq  \frac{\bQ - 2 F(A_{<i})}{P} \cdot \sum_{j \in [m] } y^*_j c_j  = \bQ - 2 F(A_{<i}),
\end{align*}
where the last equality is because by definition $\sum_{j \in [m] } y^*_j c_j = P$. By above equation, $\bQ - 2 F(A_{<i})  < \bQ - 2 F(A_{<i})$, a contradiction. 
\Qed{Lemma~\ref{lem:greedy-progress}}

\end{proof}

\begin{proof}[Proof of Lemma~\ref{lem:bound-A}]
Fix any iteration $i$ in the algorithm where $F(A_{<i}) \leq \bQ/3$. By Lemma~\ref{lem:greedy-progress}, 
\begin{align}
	F_{A_{<i}}(j_i) \Geq{Lemma~\ref{lem:greedy-progress}} c_{j_i} \cdot \frac{\bQ - 2 F(A_{<i})}{P} \geq c_{j_i} \cdot \frac{\bQ}{3P}. \label{eq:coverage-rate}
\end{align}
\noindent
Let $k$ be the first index where $F_{A_{<k}} < \bQ/3$ but $F_{A_{<k+1}} \geq \bQ/3$ (i.e., the iteration the algorithm terminates). Note that $\cost{\rX_A} = \sum_{i=1}^{k}c_{j_i}$. 
We start by bounding the first $k-1$ terms in $\cost{\rX_A}$: 
\begin{align*}
	\bQ/3 > F({A_{<k}}) &= \sum_{i=1}^{k-1} F_{A_{<i}} (j_i) \Geq{Eq~(\ref{eq:coverage-rate})} \sum_{i=1}^{k-1} c_{j_i} \cdot \frac{\bQ}{3P}  \\
	&\implies \sum_{i=1}^{k-1} c_{j_i} < P. 
\end{align*}
Now consider the last term in $\cost{A}$, i.e., $c_{j_k}$. Again, by Lemma~\ref{lem:greedy-progress}, we have, 
\begin{align*}
	c_{j_k} &\Leq{Lemma~\ref{lem:greedy-progress}} \frac{F_{A_{<k}}(j_k) \cdot P}{\bQ - 2 F(A_{<k})} \leq  \frac{\paren{\bQ - F(A_{<k})} \cdot P}{\bQ - 2 F(A_{<k})} \leq 2P,
\end{align*}
using the fact that $F(A_{<k}) < \bQ/3$. As such, $\cost{\rX_A} \leq 3P$ finalizing the proof. 
\Qed{Lemma~\ref{lem:bound-A}}

\end{proof}

Theorem~\ref{thm:detsel} now follows immediately from Lemma~\ref{lem:bound-A} and Lemma~\ref{lem:lower-bound-opt} as $P \leq \Ex\bracket{\cost{\OPT}}$.

\subsection{Proof of Theorem~\ref{thm:select}}

We use the algorithm in Theorem~\ref{thm:detsel} to present the following algorithm for reducing the expected deficit of any given set $\rS$ in Theorem~\ref{thm:select}. 

\begin{tbox}
\label{algo:select}
{$\Select(\rX, g, \rS, \alpha)$. Given a collection of indices $\X$, a monotone submodular function $g$ with $g(X) = Q_g$ for every $X \sim \X$, collection of items $\S$ with expected deficit $\Delta = \E[Q_g -g(\S)]$, picks a set $\rT$ of items to decrease the expected deficit.} 
\begin{enumerate}
\item Let $\Xi := 6\alpha$. 
\item \textbf{For} $i = 1,  \cdots , \Xi$ \textbf{do}: 
\begin{enumerate}
\item Sample a realization $S_i \sim \S$. 
\item $\rT_i \leftarrow \nagreedy(\rX \setminus \rS, g_{S_i},\Delta(S_i))$ (recall that $\Delta(S_i) = Q_g - g(S_i)$).
\end{enumerate}
\item \textbf{Return} all items in the sets $\rT:= \rT_1 \cup \rT_2 \cdots \cup \rT_{\Xi}$.
\end{enumerate}
\end{tbox}

The $\Select$ algorithm repeatedly calls the $\nagreedy$ algorithm for samples drawn 
from realizations of the set $\rS$.
By Fact~\ref{fact:marginal-submodular}, for any realization $S_i$ of $\rS$, $g_{S_i}(\cdot)$ is also a monotone submodular function. 
Moreover, by the assumption that $g(X) = Q_g$ always, we have that $g_{S_i}(X \setminus S_i) = Q_g - f(S_i)$ always as well. Hence, the parameters given to function $\nagreedy$ in $\Select$ are valid. 
 
We first bound the expected cost of \Select.

\begin{claim}\label{clm:select-cost}
	$\Ex\bracket{\cost{\rT}} = O(\alpha) \cdot \Ex\bracket{\cost{\OPT}}$. 
\end{claim}
\begin{proof}
	Cost of $\rT$ is the cost of the sets $\rT_1,\ldots,\rT_{\Xi}$ chosen by $\nagreedy$ on $g_{S_i}$ for each of the $\Xi$ realizations of $\rS$. By Theorem~\ref{thm:detsel}, we can bound the cost of each $\rT_i$ 
	using $\OPT$ conditioned on realization $S_i$ for $\rS$ (as we consider $g_{S_i}$). As such, 
	\begin{align*}
		\Ex\bracket{\cost{\rT}} 
		&  = \sum_{i=1}^{\Xi} \Ex_{S_i \sim \rS}\bracket{\cost{\rT_i}} \\
		& \Leq{(a)} \sum_{i=1}^{\Xi} \Ex_{S_i \sim \rS}\bracket{3 \cdot \Ex_{X}\bracket{\cost{\OPT(X)} \mid \rS = S_i}} \\
		&= \sum_{i=1}^{\Xi} 3 \cdot \Ex_{S_i \sim \rS}\Ex_{X \sim \rX \mid S_i}\bracket{\cost{\OPT(X)}} \\
		&= 3\Xi \cdot \Ex_{X}\bracket{\cost{\OPT(X)}}
			\,,
	\end{align*}
	where the inequality $(a)$ follows from \Thm{thm:detsel} because even though the $\OPT$ used in \Thm{thm:detsel} is an optimal algorithm on the problem instance $(\tilde{Q}, \X\setminus \S)$, but the cost of $\Ex_{X}\bracket{\cost{\OPT(X)} \mid \rS = S_i}$ can only be larger than 
the cost of $\OPT$ on the instance $(\tilde{Q}, \X\setminus \S)$.
	The bound now follow from the value of $\Xi = 6\alpha$. 
\end{proof}

We now prove that the expected deficit of $f(\rS \cup \rT)$ is dropped by at least a $\Delta/6$ factor. The following lemma is at the heart of the proof. 
\begin{lemma}\label{thm:select-deficit}
	$\Ex\bracket{\Delta(\rS \cup \rT)} \leq 5\Delta/6$. 
\end{lemma}
\begin{proof}
We start by introducing the notation needed in the proof. 
It is useful to note that the randomness in $\rT_i$
is due to two sources: (1) the sample $S_i \sim \S$
which determines which sets are \emph{indexed} by $\rT_i$;
and (2) the randomness in the \emph{realization} of the sets indexed by $\rT_i$.
For any realization $S$ of $\rS$, we use $\rT_i(S)$ to denote the set $\rT_i$ chosen (deterministically now by $\nagreedy$)
conditioned on $\rS = S$ (this corresponds to ``fixing'' the first source of randomness above). 
We use the notation $\rT_{\leq i}$ to denote the 
collection $\rT_1 \cup \cdots \cup \rT_i$ of sets selected 
in iterations $1$ through $i$,
and $S_{\leq i}$ to denote the tuple of realizations
$(S_1, \cdots, S_i)$ (we define $\rT_{<i}$ and $S_{<i}$ analogously, where $\rT_{<1} = S_{<1} = \emptyset$).
We also denote by $\rT_{\leq i} (S_{\leq i})$ the sets selected 
in iterations $1$ to $i$ given $S_{\leq i}$.

Consider any $i \in [\Xi]$. For a realization $S_i \sim \rS$, we are computing $\nagreedy$ on $g_{S_i}$ with parameter $\bQ = \Delta(S_i)$. As such, by Theorem~\ref{thm:detsel},
for the set $\rT_i(S_i)$ returned, we have $\Ex_{X}\bracket{g_{S_i}(\rT_i(S_i))} \geq \bQ/3 = \Delta(S_i)/3$. Consequently, 
\begin{align}
\label{eqn:detsel}
\Ex_{S_i \sim \S} \Ex_{X}[ g_{S_i} \big(\rT_i(S_i)\big)  ] \geq \Ex_{S_i \sim \S} [\frac{\Delta(S_i)}{3}] = \frac{\Delta}{3}.
\end{align}

We now use this equation to argue that adding each set $\rT_i$ can decrease the expected deficit. Before that, let us briefly touch upon the difficulty in proving this statement and the intuition behind the proof. 
In $\Select$, we first pick a realization $S_i$ of $\rS$ and then add ``enough'' sets to $\rT_i$ to (almost) cover the deficit \emph{introduced by} $S_i$. This corresponds to Eq~(\ref{eqn:detsel}). 
However, our goal is to decrease the \emph{expected} deficit of $\rS$ (not a deficit of a single realization). As such, the quantity of interest is in fact the following instead: 
\begin{align}
	\Ex_X\bracket{g_{\rS}(\rT_i)} = \Ex_{S_i \sim \rS} \Ex_{S'_i \sim \rS} \Ex_{X}\bracket{g_{S'_i}(\rT_i(S_i))}, \label{eq:compare}
\end{align}
i.e., the marginal contribution of $\rT_i(S_i)$ (chosen by picking a set $S_i$) to a ``typical'' set $S'_i \sim \rS_i$ (not exactly the set $S_i$). 
The set $\rT_i$ we picked in this step is \emph{not necessarily} covering the deficit introduced by $S'_i$ as well (in the context of the stochastic set cover problem, think of $S_i$ and $S'_i$ as covering a completely different set 
of elements and $\rT_i$ being a deterministic set covering $U \setminus S_i$). As such, it is not at all clear that picking the set $\rT_i$ should make  ``any progress'' towards reducing the expected deficit. 

The way we get around this difficulty is to additionally consider the marginal contribution of the sets $\rT_1,\ldots,\rT_{\Xi}$ to \emph{each other}. If $\rT_1$ cannot decrease the expected deficit of most realizations $S$ chosen 
from $\rS$, then this means that by picking another set $\rT_{2}(S)$ (for a realization $S$ of $\rS$), the set $\rT_1 \cup \rT_2$ needs to have a coverage larger than both $\rT_1$ and $\rT_2$ individually (in the context of the set cover problem, 
since $\rT_1$ is ``useless'' in covering deficit created by $S$, and $\rT_2$ can cover this deficit, this means that $\rT_1$ and $\rT_2$ should not have many elements in common typically). We formalize this intuition in the 
following claim (compare Eq~(\ref{eqn:term-condition}) in this claim with Eq~(\ref{eq:compare})). 

\begin{claim}
\label{clm:increase-coverage}
Suppose at the start of iteration $i$ the following holds
\begin{equation}
\label{eqn:term-condition}
\Ex_{S_i \sim \S} \Ex_{S_{<i} \sim \S} \Ex_X[  g_{S_i}\big( \rT_{<i}(S_{<i})  \big)]  < \frac{\Delta}{6}.
\end{equation}
Then,
\[
\Ex_{S_{\leq i} \sim \S} \Ex_X\bracket{g_{T_{<i}(S_{<i})}(T_i(S_i))}  > \frac{\Delta}{6}.
\]

\end{claim}
\begin{proof}
By subtracting \Eqn{eqn:term-condition} from \Eqn{eqn:detsel}, and using linearity of expectation we get that:
\begin{align}
\frac{\Delta}{6} & <  \Ex_{S_i \sim \S } \Ex_{S_{< i} \sim \S }\Ex_X [ g_{S_i} \big(\rT_i(S_i)\big)  -  g_{S_i}\big( \rT_{< i}(S_{ < i}) \big) ]	\notag \\
			& = \Ex_{S_i \sim \S } \Ex_{S_{< i} \sim \S }\Ex_X [ g\big(\rT_i(S_i) \cup S_i \big)  -  g\big( \rT_{< i}(S_{< i}) \cup S_i \big) ] \notag \\
		& \leq \Ex_{S_i \sim \S } \Ex_{S_{< i} \sim \S }\Ex_X [ g \big(\rT_{\leq i}(S_{\leq i}) \cup S_i \big)  -  g\big( \rT_{< i}(S_{<  i}) \cup S_i \big) ] \tag{by monotonicity} \\
		& \leq \Ex_{S_i \sim \S} \Ex_{S_{< i} \sim \S }\Ex_X [ g \big(\rT_{\leq i}(S_{\leq i}) \big)  -  g \big( \rT_{< i}(S_{< i}) \big) ] \tag{by  submodularity as $T_{< i}(S_{< i}) \subseteq T_{\leq i}(S_{\leq i})$} \\
		&= \Ex_{S_{\leq i} \sim \S }\Ex_X \bracket{g_{T_{<i}(S_{<i})}(T_i(S_i))},
\end{align}
finalizing the proof. 
\Qed{Claim~\ref{clm:increase-coverage}}

\end{proof}

Suppose towards a contradiction that $\Ex\bracket{\Delta(\rS \cup \rT)} > 5\Delta/6$. This implies that, 
\begin{align*}
	5\Delta/6 &< \Ex\bracket{Q_g - g(\rS \cup \rT)} = \Ex\bracket{Q_g - g(\rS) - g_{\rS}(\rT)} \\
	&\implies \Ex_{S \sim \rS}\Ex_{X}\bracket{g_S(\rT)} < \Delta/6. 
\end{align*}
By monotonicity of $f$ and since $\rT = \rT_1 \cup \ldots \cup \rT_\Xi$, this implies that for all $i \in [\Xi]$, 
\begin{align*}
	\Delta/6 > \Ex_{S \sim \rS}\Ex_{X}\bracket{g_S(\rT_{\leq i})} = \Ex_{S_i \sim \S} \Ex_{S_{<i} \sim \S} \Ex_X[  g_{S_i}\big( \rT_{<i}(S_{<i})  \big)]. 
\end{align*}
Hence, we can apply Claim~\ref{clm:increase-coverage} to obtain that for any $i \in [\Xi]$: 
\[
\Ex_{S_{\leq i} \sim \S} \Ex_X\bracket{g_{T_{<i}(S_{<i})}(T_i(S_i))}  > \frac{\Delta}{6}.
\]
As such, by linearity of expectation and above equation,  
\begin{align*}
	\Ex_{S_{\leq \Xi} \sim \S} \Ex_X\bracket{g(\rT(S_{\leq \Xi}))} 
	&= \sum_{i=1}^{\Xi} \Ex_{S_{\leq i} \sim \S} \Ex_X\bracket{g_{T_{<i}(S_{<i})}(\rT_i(S_i))}  \\
	&> \Xi \cdot \frac{\Delta}{6} = 6 \alpha \cdot \frac{\Delta}{6} \\
	& \geq Q_g = \Ex_{X}[g(\X)]. 
\end{align*}
where the last inequality follows due to the condition that $\alpha \geq Q_g/\Delta$. The above is a contradiction as $\rT \subseteq \X$ and $g$ is monotone. Hence, $\Ex\bracket{\Delta(\rS \cup \rT)} \leq 5\Delta/6$, finalizing the proof. 
\Qed{Lemma~\ref{thm:select-deficit}}

\end{proof}

Theorem~\ref{thm:select} now follows immediately from Claim~\ref{clm:select-cost} and Lemma~\ref{thm:select-deficit}.

%% file: r-round-upper.tex
\section{Algorithms for the Stochastic Submodular Cover Problem}\label{sec:upper}

In this section, we present our main algorithmic result which formalizes Result~\ref{res:r-round-upper}.

\begin{theorem}\label{thm:r-round-upper}
	Let $E$ be a ground-set of items, $f:2^{E} \rightarrow \IN_+$ be a monotone submodular function with $Q := f(E)$,
	and $\rX:=\set{\rX_1,\ldots,\rX_m}$ be a collection of $m$ stochastic items with support in $E$. Let $c_i \in [C]$ be the integer-valued cost of item $\X_i$.
	For any integer $r \geq 1$, there exists an $r$-round adaptive algorithm for the stochastic submodular cover problem for function $f$ 
	and items $\X$ with expected cost $O(r \cdot Q^{1/r} \cdot \log{Q} \cdot \log (mC))$ times the expected cost of the optimal adaptive algorithm. 
\end{theorem}
\noindent
Theorem~\ref{thm:r-round-upper} immediately implies that the $r$-round adaptivity gap of the stochastic submodular cover problem is $\Ot(Q^{1/r})$. 
The rest of this section is devoted to the proof of Theorem~\ref{thm:r-round-upper}. 
\paragraph{Overview.} The underlying strategy behind our algorithm is as follows: in each round of the algorithm, reduce the \emph{deficit} of the currently realized set $T$ chosen in the previous rounds (i.e., the quantity $Q - f(T)$)
by a factor of roughly $Q^{1/r}$. This suggests that after $r$ rounds the deficit should reach zero, hence we obtain a submodular cover. 
In order to do so, the algorithm needs to specify an ordering of items \emph{without} knowing the realizations of these items in advance (i.e., non-adaptively). 
This step is itself done by running the algorithm in Theorem~\ref{thm:select} over multiple iterative \emph{phases} to reduce the deficit without knowing realization of any chosen items in this round. 
We now present our algorithm in details, starting with its main component for reducing the deficit in each round. 

\subsection{The $\Reduce$ Subroutine}\label{sec:subroutine}

Let $\rT_k$ be the items selected by the $r$-round adaptive algorithm in rounds up to (and including) $k$,
and $T_k$ be their realization. In round $k$, the algorithm creates an ordering of all the available items and sets a threshold $\tau_k := Q - Q^{(r-k)/r}$ for coverage in this round: after deciding on an ordering
of the items non-adaptively, the algorithm picks items according to this ordering one by one until the total coverage of the function reaches $\tau_k$. In this section, we design an 
algorithm, namely $\Reduce$, which returns an \emph{ordered set} 
$\rS \subseteq \X \setminus \rT_{k-1}$ in round $k$
such that items in $\rS$ are enough to reach the coverage
threshold for this round with high probability.
If there are items that are not included in $\rS$ by $\Reduce$, we will simply add them 
at the end of $\rS$ in any arbitrary order.

The input to the function $\Reduce$ in round $k$ is the set of items $\rX \setminus \rT_{k-1}$, and the function marginal $f_{T_{k-1}} $; by Fact~\ref{fact:marginal-submodular}, $f_{T_{k-1}}$ is also a monotone submodular function. The execution of $\Reduce$ is 
partitioned over $\Gamma := O(\log{(mC)})$ \emph{phases}, where in each phase, the algorithm picks a new  set of items to be added to the (ordered) set returned by it. 
The final set of items returned by $\Reduce$ are ordered in increasing order of the phases (with arbitrary ordering in each phase). 

For any phase $p \in [\Gamma]$, we define $\S_p$ as the ordered set of items selected in phase $1$ up to (and including) $p$. Let $Q_{k} := Q- f(T_{k-1})$; this is the \emph{deficit} of the set $T_{k-1}$ with respect to function $f$. 
For any set $\rS$ of items, we define the following event $\event_k(\S)$:
\begin{align}
	\event_k(\S) := \1[Q_{k} - f_{T_{k-1}}(\S) \geq Q_{k} /Q^{1/r}].
\end{align}
Intuitively speaking, $\event_k(\S)$ happens if the set of items $\S$ cannot cover most of $Q_{k}$ yet. 

In each phase, $\Reduce$ makes $\Lambda := O(\log{Q})$ calls to $\Select$ subroutine (\Thm{thm:select}). 
Each call in phase $p$ is to increase the coverage of the set $\S_{p-1}$ to eventually achieve a larger coverage in $\S_{p}$. 
Instead of passing $\S_{p-1}$ directly to $\Select$, we instead pass the set $\tS_{p-1} := \rS_{p-1} \mid \event_k(\rS_{p-1})$ which is a set of items that has the same distribution as 
$\rS_{p-1}$ conditioned on the event $\event_k(\rS_{p-1})$ (i.e., we only consider such realizations of $\rS_{p-1}$ where $\event_k(\S_{p-1})$ occurs). We show in Claim~\ref{clm:select} that the performance of $\Select$ remains 
the same in this case also (simply because in $\Select$ we only access the distribution of input sets by sampling from it and hence we can sample from $\tS_{p-1}$ instead of $\rS_{p-1}$). 
This step is required to ensure that we can indeed achieve a larger coverage with higher probability across phases as we are ``focusing'' on realizations that are ``bad'' in previous phases, i.e.,
cannot cover a large fraction of $Q_{k}$. Formally, we prove that the $\Pr\paren{\event_k(\rS_{p})} \leq 1/2 \cdot \Pr\paren{\event_{k}(\rS_{p-1})}$, hence after $\Gamma = \Theta(\log{(mC)})$ phases,
the probability of this bad event reduces to $1/(mC)^{O(1)}$ and we can move on to the next round. We present the pseudo-code of $\Reduce$ algorithm below.

\begin{tbox}
\label{algo:reduce}
$\Reduce(\X,f_{T_{k-1}})$: Given a set $\X$ of items and a monotone submodular function $f_{T_{k-1}}$, outputs an ordered set of items $\S$ to be used in round $k$ of the $r$-round adaptive algorithm.
\begin{enumerate}
\item \textbf{Initialize:}  Set $\Lambda \leftarrow 12 \log(Q) $, and $\Gamma \leftarrow 2 \log{(mC)}$ .   

\item Set $\S_0 \leftarrow \emptyset$.
\item \textbf{For} phases $p = 1,  \cdots ,  \Gamma$ \textbf{do:} 
\begin{enumerate}
\item Set $\R_0 \leftarrow \emptyset$ and let $\tS_{p-1} := \S_{p-1} \mid \event_k(\S_{p-1})$. 
\item \textbf{For} iterations $i = 1, \cdots , \Lambda$ \textbf{do:}
\begin{enumerate}
\item $\R_i \leftarrow  \R_{i-1} \cup \Select( \rX \setminus \{\R_{i-1} \cup \S_{p-1}\}, f_{T_{k-1}}, \R_{i-1} \cup \tS_{p-1} , 2Q^{1/r}) $.
\end{enumerate}

\item $\S_p \leftarrow \S_{p-1} \cup \R_\Lambda$.

\end{enumerate}

\item \textbf{Return} the set $\S_{\Gamma}$, ordered according to the order in which items were added to $\S_{\Gamma}$.


\end{enumerate}
\end{tbox}

Before analyzing $\Reduce$ we need the following straightforward extension of Theorem~\ref{thm:select}.

\begin{claim}[Extension of Theorem~\ref{thm:select}]
\label{clm:select}
Let $f_T$ be any monotone submodular function, for some $T \subseteq E$, such that $Q' := Q - f(T)$.
Let $\rS \subseteq \rX$ be any subset of items, and $\event$ be an event which is a function of $\S$ and $\tS := \S \mid \event$. 
Let $\Delta := \Ex[Q' - f_T(\tS )]$, then $\Select$, given parameter $\alpha \geq Q'/\Delta$,
and $6 \alpha$ samples from $\tS$, outputs a set $\R \subseteq \rX \setminus \rS$ such that cost of $\R$ is $O(\alpha) \cdot \Ex\bracket{\cost{\OPT} \vert \event}$ in expectation over the randomness of the algorithm and 
$\Ex\bracket{Q' - f_T(\tS \cup \R)} \leq 5\Delta/6$ over the randomness of the algorithm and realizations of $\tS$ and $\R$. 
\end{claim}
\noindent
Claim~\ref{clm:select} can be proven as follows: in $\Select$ we only need samples from the distribution $\rS$, hence by sampling from the distribution of $\tS$ instead we obtain the same result conditioned on event $\event$. One should be careful though, as the items in $\tS$ are no longer independent 
due to the conditioning on $\event$. However, $\Select$ does not require
independence between items in $\S$ and we can simply use $\tS$
instead of $\rS$.

We start by bounding the cost of the sets returned by $\Reduce$ in each phase. Note that not all these sets are going to be chosen by the $r$-round algorithm in round $k$ (as we may cover $\tau_k$ before reaching these sets and move on to next round)
and hence this cost is \emph{not} a lower bound on cost of the $r$-round algorithm. 

\begin{claim}\label{clm:reduce-cost}
	For any $p \in [\Gamma]$, $\Ex\bracket{\cost{\rS_p \setminus \rS_{p-1}}} = O(Q^{1/r} \cdot \log{Q}) \cdot \Ex[\c(\OPT) | \event_k(\S_{p-1})]$. 
\end{claim}
\begin{proof}
We call $\Select$ with the parameter $2Q^{1/r}$ for $O(\log{Q})$ iterations. By \Clm{clm:select}, cost of each iteration of phase $p$ is at most
$O(Q^{1/r})$ times the expected cost of $\OPT$ conditioned on $\event_k(\S_{p-1})$.
Hence, total cost of phase $p$ is $\Ex[\cost{\rS_p \setminus \rS_{p-1}}] = O(Q^{1/r} \cdot \log{Q}) \cdot \Ex[\cost{\OPT} \vert{ \event_k(\S_{p-1})}]$. 
\end{proof}
We now prove the main property of the $\Reduce$ subroutine, i.e., that the sets returned by it can cover the required threshold $\tau_k$ with high probability. 

\begin{lemma}
\label{lem:reduce-r-round}
Suppose $\rS_{\Gamma} := \Reduce(\rX,f_{T_{k-1}})$. Then,   
\[
\Pr(\event_k(\S_{\Gamma})) \leq 1/(mC)^2, 
\]
with respect to the randomness of the algorithm and the realizations of $\S_\Gamma$.
\end{lemma}

\begin{proof}
We prove that the probability 
of the event $\event_k(\S_p)$
decreases after each phase $p$ by a constant factor.
Fix a phase $p \in [\Gamma]$. For a realization $S$ we define deficit $\Delta(S) =  Q_{k}- f_{T_{k-1}}(S)$.
Recall that $\R_i$ is the set of items picked up to (and including) iteration $i$ in phase $p$
on calls to $\Select$ with parameter $\alpha = 2Q^{1/r}$.
By \Clm{clm:select} we know that each iteration reduces the expected deficit 
by a constant factor. More formally, fix an $\R_{i-1}$ selected up to iteration $i-1$.
If $\Ex  \bracket{\Delta(\R_{i-1} \cup \S_{p-1}) \vert  \event_{p-1}} \geq Q_{k} /2Q^{1/r}$,
then the condition of \Clm{clm:select} that $\alpha \geq Q'/\Delta$ is satisfied with 
$\Delta = \Ex \bracket{\Delta(\R_{i-1} \cup \S_{p-1}) \vert  \event_{p-1}}$,
 $\alpha = 2Q^{1/r}$, and $Q' = Q_{k}$.
We then have
\begin{align*}
\Ex &\bracket{\Delta(\R_{i} \cup \S_{p-1})  \vert  \event_k(\S_{p-1})}  \\
	  & \qquad\qquad\qquad\leq \frac{5}{6} \Ex\bracket{\Delta(\R_{i-1} \cup \S_{p-1}) \vert  \event_k(\S_{p-1})},
\end{align*}
where the above expectation is also over the randomness of the $\Select$ subroutine in iteration $i$, in addition of the realization of $\R_i \cup \S_{p-1}$.
Now, we will prove that $\Lambda$ iterations are enough to drop the expected deficit 
below $Q_{k}/2Q^{1/r}$.
Suppose for a contradiction that this is not the case, i.e.\ after $\Lambda $ iterations
we have that
\begin{align}
\label{eqn:cond-progress}
\Ex [ \Delta(\R_{\Lambda} \cup \S_{p-1}) \vert  \event_k(\S_{p-1})] \geq \frac{Q_{k}}{2Q^{1/r}}
	\,.
\end{align}
Due to the fact that $f_{T_{k-1}}$ is a monotone function, we have
\begin{align*}
\Ex [ \Delta(\R_{i} \cup \S_{p-1}) \vert  \event_k(\S_{p-1})] \geq \Ex [ \Delta(\R_{\Lambda} \cup \S_{p-1}) \vert  \event_k(\S_{p-1})]
	\,,
\end{align*}
for all $\R_i$.
Then using \Eqn{eqn:cond-progress} and the above equation, we can observe that the condition of \Clm{clm:select} that $\Ex  \bracket{\Delta(\R_{i} \cup \S_{p-1}) \vert  \event_k(\S_{p-1})} \geq Q_{k} /2Q^{1/r}$ is satisfied for every $\R_{i}$.
This implies that after $\Lambda$ iterations the expected deficit can be written as
\begin{align}
\Ex  [ \Delta(\R_\Lambda \cup \S_{p-1}) \vert \event_k(\S_{p-1})]  
&\leq \paren{\frac{5}{6}}^{\Lambda} \cdot \Ex  [ \Delta(\S_{p-1}) \vert \event_k(\S_{p-1})] \nonumber \\
&\leq \paren{\frac{5}{6}}^{12\log{Q}} \cdot Q_{k}   \qquad\qquad\qquad \tag{Recall that $\Lambda = 12 \log Q$} \nonumber\\
&<\frac{Q_{k}}{2Q}  \leq \frac{Q_{k}}{2Q^{1/r}} \label{eqn:contradict}
	\,.
\end{align}
\Eqn{eqn:cond-progress} and \Eqn{eqn:contradict} lead to a contradiction. Hence, we will have that 
\begin{align*}
\Ex [ \Delta(\S_p) \vert  \event_k(\S_{p-1})]  = \Ex  [ \Delta(\R_\Lambda \cup \S_{p-1}) \vert \event_k(\S_{p-1})] < \frac{Q_{k}}{2Q^{1/r}}
	\,.
\end{align*}
%
where again the expectation is over the randomness of the $\Select$ subroutine.
Now, using Markov's inequality we have that
\begin{align}
\label{eqn:reduce-first}
\Pr\Big(\event_k(\S_{p})\, &\Big\vert \, \event_k(\S_{p-1})\Big) = \Pr\Big(\Delta(\S_p)  \geq \frac{Q_{k}}{Q^{1/r}} \, \Big\vert \, \event_k(\S_{p-1}) \Big) \leq \frac{1}{2} 
\,,
\end{align}
where the above probability is both with respect to the realizations of $\S_p$ and the coins used by the 
algorithm to select $\S_p$.
Now, we have that 
\begin{align*}
\Pr(\event_{k}(\rS_{\Gamma})) &= \Pr(\event_k(\rS_1)) \cdot \Pi_{p=2}^\Gamma \Pr\paren{\event_k(\rS_i) \mid \event_k(\rS_{i-1})} \\
&\!\!\!\!\!\Leq{Eq~(\ref{eqn:reduce-first})} \paren{\frac{1}{2}}^{\Gamma-1} \leq \frac{1}{(mC)^2},
\end{align*}
by the choice of $\Gamma = \Theta(\log{(mC)})$, which proves the desired result.
\Qed{\Lem{lem:reduce-r-round}}

\end{proof}

\subsection{The $r$-Round Adaptive Algorithm}

We are now ready to present our $r$-round algorithm which is based on successive applications of the $\Reduce$ subroutine.

\begin{tbox}
\label{algo:adaptive-submodular}
$\KAdapt(\X, f, Q)$: Given a set of items $\X$, a monotone submodular function $f$, and the desired coverage value $Q$, outputs a set $\rT$ such that its realization $T$ is feasible. 
\begin{enumerate}
\item \textbf{Initialize:} Set $\rT_0 \leftarrow \emptyset, T_0 \leftarrow \emptyset$
\item \textbf{For} $k = 1,2,  \cdots , r$ \textbf{do:} 
\begin{enumerate}
\item Set threshold $\tau_k \leftarrow Q - Q^{(r-k)/r}$
\item $\rT \leftarrow \Reduce(\X \setminus \rT_{k-1}, f_{T_{k-1}}) $
\item  Add the remaining items $\X \setminus (\rT \cup \rT_{k-1})$ at the end of $\rT$ in any arbitrary order.

\item Observe the realizations $T'$ of the set of items $\rT '\subseteq \rT$ selected by running through the ordered set $\rT$ until a total coverage of $\tau_k$ is reached, i.e.\ $f(T_{k-1} \cup T') \geq \tau_k$
\item $\rT_k \leftarrow  \rT' \cup \rT_{k-1} $ and $T_{k} \leftarrow  T' \cup T_{k-1} $
\end{enumerate} 

\item \textbf{Return } $\rT_{r}$ with realization $T_r$ as the final answer. 
\end{enumerate}
\end{tbox}

We are now ready to prove \Thm{thm:r-round-upper} by analyzing the above algorithm.
The overall plan is to bound the cost of each round of the $r$-round algorithm.
In each round the algorithm selects an ordering returned by a call 
to $\Reduce$ and adds the remaining items at the 
end of this ordering. 
As argued earlier, not all the sets in the ordering
are going to be chosen by the $r$-round algorithm in round $k$.
We will use \Clm{clm:reduce-cost} and 
Lemma~\ref{lem:reduce-r-round} to bound the
expected cost of the items selected from the ordering in round $k$
in terms of the expected cost of $\OPT$.
%
In order to do so, we first lower bound the cost of $\OPT$.

\begin{claim}
\label{clm:bound-opt}
For any (possibly randomly chosen) collection $\S \subseteq \X$, and any event $\event$
which is a function of $\S$,
the expected cost of $\OPT$ can be lower bounded as
\[
\Ex[\cost{\OPT}] \geq \Pr(\event) \cdot \Ex[\cost{\OPT} \vert{\event}].
\]
\end{claim}
\begin{proof}
The expected cost of $\OPT$ can be written as
\begin{align*}
\Ex[\cost{\OPT}] &= \Pr(\event) \cdot \Ex[\cost{\OPT} \vert{\event}]  + \Pr(\neg \event) \cdot \Ex[\cost{\OPT} \vert{\neg \event}]  \\
		     &\geq \Pr(\event) \cdot \Ex[\cost{\OPT} \vert{\event}]
		     		     	\,.
\end{align*}
Note that the above also holds even if the collection $\S$ is itself randomly chosen.
\Qed{\Lem{clm:bound-opt}}

\end{proof}

We now prove the lemma bounding the expected cost of each round of $\KAdapt$.
We will define the notation $\cost{\Round_k}$ to be 
the total cost of all the items added to the feasible 
set in round $k$. More formally,
\[
\cost{\Round_k} := \cost{\rT_k \setminus \rT_{k-1}}
	\,.
\]
Now, we will provide a bound on $\Ex[\cost{\Round_k}]$.

\begin{lemma}
\label{lem:round-cost}
For any round $k\leq r$, given $T_{k-1}$, the expected cost 
paid by the $\KAdapt$ algorithm in round $k$ can bounded as
\begin{align*}
\Ex[&\cost{\Round_k} \vert T_{k-1}]  \leq  O(Q^{1/r} \log(Q) \log(mC)) \cdot \Ex[\c(\OPT) \vert T_{k-1}] 
	\,.
\end{align*}
%
\end{lemma}

\begin{proof}
Recall that in round $k$ we call 
$\Reduce$ with parameter $f_{T_{k-1}} = f_{T_{k-1}}$ such that $Q_{k} = Q -  f(T_{k-1})$.
Also, recall that in phase $p$, $\Reduce$ adds items 
$\S_p \setminus \S_{p-1}$
to the ordering $\S_\Gamma$ returned by it.
Using \Clm{clm:reduce-cost} we have that
\begin{align*}
\Ex[&\cost{\S_p \setminus \S_{p-1}} \vert T_{k-1} ]  = O(Q^{1/r} \cdot \log{Q}) \cdot \Ex[\c(\OPT) \vert T_{k-1}, \event_k(\S_{p-1})]
		\,.
\end{align*}
Also, recall that while running through the ordered set
of round $k$ we select items 
from $\S_p \setminus \S_{p-1}$ only
if the realization is such that the items in $\S_{p-1}$ are not able to 
reach the required coverage threshold $\tau_k$.
More formally, we only pay for the cost of items in $\S_p \setminus \S_{p-1}$
when the event $\event_k(\S_{p-1})$ occurs.
Hence, we will pay the cost of phase $p$
items with probability $\Pr(\event_k(\S_{p-1}))$.
Also, in the case that all the items $\S_{\Gamma}$ returned by $\Reduce$ are not able 
to reach the required coverage threshold, we trivially bound the cost by 
$mC$. 
Since, $Q_k \leq Q^{(r-k+1)/r}$, 
this event happens with probability at most $\Pr(\event_k(\S_\Gamma))$ which is upper bounded by
$1/(mC)^2$ using \Lem{lem:reduce-r-round}.
Combining all this, we have that, given $T_{k-1}$,
\begin{align*}
&\Ex[\cost{\Round_k} \vert T_{k-1} ]  \\
&\qquad \leq \sum_{p=1}^{\Gamma}  \Pr\paren{\event_k(\S_{p-1})} \cdot  \Ex[\cost{\S_p \setminus \S_{p-1}} \vert T_{k-1}] + \Pr\paren{\event_k(\S_{\Gamma})} \cdot mC \\
& \qquad \Leq{\Clm{clm:reduce-cost}} \sum_{p=1}^{\Gamma} \Pr\paren{\event_k(\S_{p-1})} \cdot O\paren{Q^{1/r} \log (Q)} \cdot \Ex[\c(\OPT) \vert T_{k-1}, \event_k(\S_{p-1})] + \Pr\paren{\event_k(\S_{\Gamma})} \cdot mC
  \nonumber	\\
		& \qquad \Leq{\Clm{clm:bound-opt}} O\paren{Q^{1/r} \log (Q) \log (mC)} \Ex\bracket{\cost{\OPT} \vert T_{k-1}} + \Pr\paren{\event_k(\S_{\Gamma})} \cdot mC \\
		& \qquad \Leq{\Lem{lem:reduce-r-round}} O\paren{Q^{1/r} \log (Q) \log (mC)} \Ex\bracket{\cost{\OPT} \vert T_{k-1}} + \frac{1}{(mC)^2} \cdot mC \\
		&  \qquad =  O\paren{Q^{1/r} \log (Q) \log (mC)} \Ex\bracket{\cost{\OPT} \vert T_{k-1}}
		\,.
\end{align*}

\Qed{\Lem{lem:round-cost}}

\end{proof}

We are now ready to prove \Thm{thm:r-round-upper} which uses the above lemma
to give a combined bound on the cost of all the rounds.
\begin{proof}(of \Thm{thm:r-round-upper})
We will first divide the $\cost{\KAdapt}$ into the cost of each round.
\begin{align}
\label{eqn:cost-rounds}
\Ex[ \cost{\KAdapt}] &= \sum_{k=1}^r \Ex[\cost{\Round_k}] 
	\,,
\end{align}
where recall that $\cost{\Round_k} := \cost{\rT_k \setminus \rT_{k-1}}$
and $\rT_k$ are the items selected up to (and including) round $k$. Let $T_k$ be the realization of $\rT_k$. We first need to understand that there are two sources of randomness-- 1) due to the coins used by the algorithm to sample the realizations; 2) due to stochastic nature of items. 
We will first fix the randomness due to the coins used by the algorithm for sampling.
Once we fix the realization of coins used by the algorithm, the only randomness in the algorithm is 
due to the stochastic nature of items.
Then for any $k \leq r$, given a fixed realization of coins in rounds up to $k-1$, we have
\begin{align*}
\Ex[  \c(\Round_k)] &\Leq{\Lem{lem:round-cost}} O\paren{Q^{1/r} \log (Q) \log (mC)} \cdot \Ex_{T_{k-1} \sim \rT_{k-1}} \Ex [\cost{\OPT}\vert {T_{k-1}} ]  \\
		& = O\paren{Q^{1/r} \log (Q) \log (mC)}   \Ex[\c(\OPT)]
	\,,
\end{align*}
where the last equality is due to the fact that once we fix the randomness due to coins up to round $k-1$, then
the realizations $T_{k-1}$ form a partition over the space of all realizations $X$.
Since the choice of coins was arbitrary, we have that $\Ex[\c(\Round_k) ]\leq \Ot(Q^{1/r}) \c(\OPT)$.

Then, using \Eqn{eqn:cost-rounds} and the above, the total cost can be bounded as
\begin{align*}
&\cost{\KAdapt} =   O\paren{r Q^{1/r} \log (Q) \log (mC)} \Ex[\c(\OPT) ]
	\,. 
\end{align*}
\Qed{\Thm{thm:r-round-upper}}

\end{proof}

\begin{remark}
	We can implement the $r$-round algorithm in polynomial time as long as the costs are polynomially bounded, i.e., achieve a pseudo-polynomial time algorithm.
	Indeed, the only ``time consuming'' step of the algorithm is to sample from the conditional distribution $\S \vert \event$ for some event $\event$. 
	This is however is only needed as long as the $\Pr(\event) \geq 1/(mC)^{\Theta(1)}$. 
	Hence, one can use rejection sampling with 
	the total running time bounded by $\poly(QmC)$ to implement this step.
	The probability that we do not get the required number of
	 samples from the event $\event$ with
	$\Pr(\event)  \geq 1/mC$ after $\poly(QmC)$ trials is negligible, 
	and we can pay for the cost in case this bad event happens.
\end{remark}


%% file: r-round-lower.tex
\newcommand{\rV}{\ensuremath{\rv{V}}}

\newcommand{\esmall}[1]{\event_{\textnormal{\textsf{small}}}(#1)}
\newcommand{\ehit}[1]{\event_{\textnormal{\textsf{hit}}}(#1)}

\newcommand{\Uc}{\ensuremath{U_{cov}}}

\renewcommand{\estar}{\ensuremath{e^{\star}}}

\section{A Lower Bound for ${r}$-Round Adaptive Algorithms}
\label{sec:lower}

In this section, we prove a lower bound on the approximation ratio of any $r$-round adaptive algorithm for the submodular cover problem and formalize Result~\ref{res:r-round-lower}. 
We prove this lower bound for the stochastic set cover problem (see Example~\ref{exp:set-cover}) which is a special case of the stochastic submodular cover problem. 

\begin{theorem}\label{thm:r-round-lower}
	For any integer $r \geq 1$, any $r$-round adaptive algorithm for the stochastic set cover problem on instances with $m$ stochastic sets from a universe of size $n$ elements such that $m = n^{O(r)}$ has expected
	cost $\Omega(\frac{1}{r^3} \cdot n^{1/r})$ times the cost of the optimal adaptive algorithm. 
\end{theorem}

\noindent
Theorem~\ref{thm:r-round-lower} formalizes Result~\ref{res:r-round-lower} as by definition, $Q = n$ in the stochastic set cover problem. 

\paragraph{Overview.} Consider first an instance of the stochastic set cover problem which was used in  \cite{GoemansVo06} for proving a $1$-round adaptivity gap. There exists a single stochastic set, say $\rT$, which realizes to $U \setminus \set{\estar}$
for $\estar$ chosen uniformly at random from $U$ (support of $\rT$ has $n$ sets). The remaining sets in this instance are $n$ singleton sets that each deterministically  realize to some unique element $e \in U$.
Solving such an instance adaptively with just two sets, and indeed even in two rounds of adaptivity, is trivial: choose the set $\rT$ and observe its realization in the first round; next choose the singleton set that covers $\estar$. However,
consider any non-adaptive algorithm for this problem: even though it is obvious that the set $\rT$ needs to be the first set in the ordering returned by the algorithm, there is no ``good'' choice for the ordering of the remaining sets 
as the algorithm is oblivious to the identity of $\estar$ at this point. It is then fairly easy to see that no matter what ordering the non-adaptive algorithm chooses, in expectation $\Omega(n)$ sets needs to be picked before it could
cover $\estar$ and hence the universe $U$. An adaptivity gap of $\Omega(n)$  now follows easily from this argument. 

Our main contribution in this section is to design a family of instances in this spirit that allows us to extend the above argument to $r$-round adaptive algorithms. Roughly speaking, these instances are constructed in a way that
at the beginning of each round, the algorithm has access to a set that covers a ``large'' portion of the remaining universe ``randomly'', but since the realization of this set is not known to the algorithm, unless it picks many more sets,
it would not be able to also cover the ``remainder of universe'' (left out by the realization of the aforementioned set). Morally speaking, this corresponds to replacing the set $\set{\estar}$ with larger subsets of $U$ in the above argument 
and then recurse on each subset individually.

The rest of this section is devoted to the proof of Theorem~\ref{thm:r-round-lower}. 
We start by introducing an algebraic construction of a set-system, named an \emph{edifice}, due to Chakrabarti and Wirth~\cite{ChakrabartiW16} and use it to introduce a family of ``hard'' instances for the stochastic set cover problem.
We then prove that any algorithm with limited rounds of adaptivity on these instances necessarily incurs a large cost compared to the optimal adaptive algorithm and prove Theorem~\ref{thm:r-round-lower}. 

\subsection*{Edifice Set-System}

An edifice over a universe $U$ of $n$ items is a collection of sets in which for any two sets, either one of them is a subset of the other, or the two sets have a small intersection. 
Formally:

\begin{definition}[Edifice Set-System~\cite{ChakrabartiW16}]\label{def:edifice}
	For integers $k \leq s \leq b \leq d$, 
	a $(s,b,k,d)$-edifice $\TT$ over a universe $U$ is a \emph{complete} \emph{$d$-ary $k$-level rooted tree} together with a collection of associated sets, satisfying the following properties: 
	\begin{enumerate}[label=(\Roman*)]
		\item Each vertex $v$ in $\TT$ is associated with a set $U_v \subseteq U$ such that the set associated to the root of $\TT$ is $U$, and $U_u \subseteq U_v$ if $u$ is a \emph{child} of $v$ in $\TT$.  
		\item If $v$ is a leaf of $\TT$, then $\card{U_v} = b$. 
		\item For each leaf $u$ and each node $v$ \emph{not an ancestor} of $u$ in $\TT$, $\card{U_u \cap U_v} \leq s$. 
	\end{enumerate}
	In this definition, we say that root is at level $1$ of the tree and the leaf-vertices are at level $k$
\end{definition}

Edifices are typically interesting when the parameter $s$ is small and parameter $b$ is large compared to the size of the universe, i.e., when we 
have large sets which are almost disjoint from each other in a recursive manner suggested by the tree-structure of an edifice. 
For our purpose, we are interested in edifices with parameters $r = k \approx s$ ($r$ is the number of rounds we want to prove the lower bound for), $b \approx n^{1/k}$, and $d = n^{O(1)}$ ($n$ is the number of elements in the universe).
The existence of such edifices follows from the results in~\cite{ChakrabartiW16} (see Theorem~{3.5}; see also RND-set systems in~\cite{AssadiK18} for a similar construction), which we summarize in the following proposition. 

\begin{proposition}[\!\!\cite{ChakrabartiW16}]\label{prop:edifice-exist}
	For infinitely many integers $N$ and any integer $k \geq 1$, there exists a $(4k, N ,k, N^{2})$-edifice over a universe $U$ of size $N^{k}$. 
\end{proposition}

\subsection*{Hard Instances for Stochastic Set Cover}

Fix an integer $k \geq 1$ and a sufficiently large integer $N \geq k$ and let $U$ be a universe of size $N^{k}$ elements. Define $\TT$ as any arbitrary $(4k, N, k,N^{2})$-edifice over $U$ which
is guaranteed to exist by Proposition~\ref{prop:edifice-exist}. We define the following family of ``hard'' instances for stochastic set cover. 

\begin{tbox}
	\textbf{Family $\rX^{(k)}$:} A collection of stochastic sets over universe $U$ using edifice $\TT$.
\begin{itemize}
	\item For any vertex $u \in \TT$ and any element $e \in U$, there exists a dedicated stochastic set $\rX_u$ and $\rX_e$ in $\rX^{(k)}$, respectively, defined as follows. 
	\item For any non-leaf vertex $u \in \TT$ with child-vertices $v_1,\ldots,v_d$, the stochastic set $\rX_u$ realizes to one of the sets  $T_{u,v_1},\ldots,T_{u,v_d}$ uniformly at random where $T_{u,v_i} := U_{u} \setminus U_{v_i}$. 
	\item For any leaf vertex $u \in \TT$ with $U_u = \set{e_1,\ldots,e_N}$ (recall that $\card{U_u} = N$ be Definition~\ref{def:edifice}), the stochastic set $\rX_u$ realizes to one of the sets $T_{u,e_1},\ldots,T_{u,e_N}$ uniformly at random
	where $T_{u,e_i} := U_u \setminus \set{e_i}$. 
	\item For any element $e \in U$, $\rX_e$ deterministically realizes to the singleton set $\set{e}$. 
\end{itemize}
\end{tbox}

For any realization of $\rX^{(k)}$, we define the \emph{canonical path} of the realization as the root-to-leaf path $P = v_1,v_2,\ldots,v_k$ over the vertices of the edifice $\TT$ as follows: 
\begin{enumerate}
\item $v_1$ is the root of the tree $\TT$. 
\item For any $1 < i \leq k$, $v_i$ is the child-vertex of $v_{i-1}$ corresponding to $T_{v_{i-1},v_{i}} = \rX_{v_{i-1}}$. 
\end{enumerate}
We have the following simple claim on the cost of the optimal adaptive algorithm on the family $\rX^{(k)}$ for any integer $k \geq 1$. 

\begin{claim}\label{clm:opt-adaptive-lower}
	For any integer $k \geq 1$, the expected cost of \OPT on $\rX^{(k)}$ is at most $k+1$. 
\end{claim}
\begin{proof}
	We prove that the following algorithm has expected cost $k+1$; clearly optimal adaptive algorithm can only have a lower expected cost. 
	
	Consider the adaptive algorithm that constructs the canonical path of the underlying realization one vertex at a time: it first chooses $v_1$ which is the root of $\TT$ and add $\rX_{v_1}$ to $S$. Next, based on the realization of 
	$\rX_{v_1}$, it can determine the second vertex $v_2$ in the canonical path and adds $\rX_{v_2}$ to $S$. It continues like this until it has added all sets $\rX_{v_1},\ldots,\rX_{v_k}$ to $S$ where $P:= v_1,\ldots,v_k$ is the canonical path of the realization. 
	Finally, a realization of $\rX_{v_k}$ for a leaf $v_k$ corresponds to a set $T_{v_k,e}$ that covers all of $U_{v_k}$ (the set associated with the leaf-vertex $v_k$ in the edifice) except for a single element $e$. The algorithm then picks 
	the set $\rX_e$ which deterministically realizes to $\set{e}$. 
	
	Clearly, the number of stochastic sets picked by this algorithm is $k+1$. We argue that these sets cover the universe $U$ entirely. This is because, $\rX_{v_1}$ covers $U \setminus U_{v_2}$, $\rX_{v_2}$ covers $U_{v_2} \setminus U_{v_3}$, 
	and so on until $\rX_{v_k}$ covers $U_{v_k} \setminus \set{e}$. As such, $\rX_{v_1} \cup \ldots \cup \rX_{v_k}$ covers $U \setminus \set{e}$ and picking $\rX_e$ would cover the whole universe as $\rX_e$ always realizes to $\set{e}$. 
\end{proof}

In the remainder of this section, we prove that any $(r=)$$k$-round adaptive algorithm for stochastic set cover on $\rX^{(k)}$ should incur a cost of roughly $n^{1/k}$, hence proving Theorem~\ref{thm:r-round-lower}. 
It is worth remarking that the adaptive algorithm in Claim~\ref{clm:opt-adaptive-lower} that achieves the cost of $k+1$ requires only $k+1$ rounds of adaptivity; as such, our results are in fact proving a separation between the cost of any 
$k$-round and $k+1$-round adaptive algorithms. 

Before we move on to the proof of Theorem~\ref{thm:r-round-lower}, we prove the following crucial lemma using properties of edifice $\TT$. 

\begin{lemma}\label{lem:int-small}
	Let $U_{v_k}$ be the set associated to the $k$-th vertex $v_k$ in the canonical path of $\rX^{(k)}$ in edifice $\TT$ and $C$ be any collection of sets in $\rX^{(k)} \setminus \rX_{v_k}$. Then $\card{\bigcup_{T \in C}T \cap U_{v_k}} \leq 4\card{C} \cdot k$.
\end{lemma}
\begin{proof}
	Fix any set $T \in C$. We prove that $\card{T \cap U_{v_k}} \leq 4k$ which would immediately imply the lemma. 
	
	If $T$ is a realization of some set $\rX_e$ for some element $e \in U$, then $\card{T} = 1$ and hence the claim immediately holds. Hence, suppose that $T$ is a realization of $\rX_v$ for some vertex $v \in \TT$. 
	
	If $v$ is an ancestor of $v_k$, then $T = U_v \setminus U_v'$ where $v'$ is either another ancestor of $v_k$ or it is equal to $v_k$ itself by definition of the canonical path. In either case, by property (I) of edifices in Definition~\ref{def:edifice}, 
	$U_{v_k} \subseteq U_{v'}$ and hence $T \cap U_{v_k} = \emptyset$. 
	
	If $v$ is not an ancestor of $v_k$, then $T \subseteq U_v$ as $\rX_v \subseteq U_v$ and by property (III) of edifices in Definition~\ref{def:edifice}, $\card{U_v \cap V_{v_k}} \leq 4k$ (here parameter $s=4k$) and hence $\card{T \cap V_{v_k}} \leq 4k$,
	finalizing the proof.
\end{proof}

\subsection*{Proof of Theorem~\ref{thm:r-round-lower}}

Fix any $k \geq 1$ and a $k$-round algorithm $\alg$ for the stochastic set cover problem on instance $\rX^{(k)}$. By Yao's minimax principle~\cite{Yao79}, we can assume that 
$\alg$ is deterministic. We use $\rS_1,\ldots,\rS_k$ to denote the collections of stochastic sets chosen by the algorithm in each of its $k$ adaptivity rounds. 
We further use the random variables $\rV_1,\ldots,\rV_k$ to denote the vertices on the canonical path of $\rX^{(k)}$ (note that $\rV_1$ is always root of the edifice $\TT$). 

Let $d := N^{2}$ denote the number of children any non-leaf vertex has in $\TT$. For any $i \in [k-1]$ we define the following two events: 

\begin{tbox}
\begin{center}
\textbf{Event $\bm{\esmall{i}}$}
\end{center}
\vspace{-22pt}
\begin{quote}
	The collection $\rS_i$ chosen by $\alg$ in round $i$ has size $\card{\rS_i} \leq N/8k$.
\end{quote}
\end{tbox}

\noindent The event $\esmall{i}$ is only a function of the realizations of first $i-1$ sets $\rS_1,\ldots,\rS_{i-1}$ chosen by $\alg$ in the first $i-1$ rounds plus the sets visited in round $i$ and their realizations before reaching the threshold
fixed by the algorithm to stop the round. 

\begin{tbox}
\begin{center}
\textbf{Event $\bm{\ehit{i}}$}
\end{center}
\vspace{-22pt}
\begin{quote}
	The collection $\rS_i$ chosen by $\alg$ in round $i$ contains no set $\rX_u$ where $u$ is a descendant of $v_{i+1} = \rV_{i+1}$, i.e., the $(i+1)$-th vertex in the canonical path of $\rX^{(k)}$
\end{quote}
\end{tbox}
\noindent
The event $\ehit{i}$ is also only a function of the realizations of the first $i-1$ sets $\rS_1,\ldots,\rS_{i-1},\rS_i$, as well as $\rV_{1},\ldots,\rV_{i+1}$. 

The following claim implies that event $\esmall{i}$ is most likely to result in $\ehit{i}$ as well. 

\begin{claim}\label{clm:event-v-i}
	For any $i \in [k-1]$, $\Pr\paren{\ehit{i} \mid \esmall{1},\ldots,\esmall{i},\ehit{1},\ldots,\ehit{i-1}} \geq 1-\frac{1}{2k}$.
\end{claim}
\begin{proof}
	Let $v_1,\ldots,v_{i}$ be the first $i$ vertices on the canonical path of $\rX^{(k)}$. By definition of events $\ehit{1},\ldots,\ehit{i-1}$, and since $v_{i}$ is a descendent of all $v_{1},\ldots,v_{i-1}$ by definition, we know 
	that no set $\rX_v$ belong to $\rS_1,\ldots,\rS_{i-1}$ for any descendent $v$ of $v_{i}$. In particular, $\rX_{v_i}$ has not been chosen in $\rS_1,\ldots,\rS_{i-1}$ and hence its distribution conditioned on $\rS_1,\ldots,\rS_{i-1}$ is still the same
	distribution as before. As such, the $(i+1)$-vertex of the canonical path of $\rX^{(k)}$, i.e., $v_{i+1}$ is still chosen uniformly at random over the child-vertices of $v_i$, even conditioned on the realizations of $\rS_1,\ldots,\rS_{i-1}$. On the other hand,
	conditioned on realizations of $\rS_1,\ldots,\rS_{i-1}$, the ordering for set $\rS_i$ chosen by $\alg$ is determined deterministically. Let $\tS$ be the set of first $N/8k$ (as in event $\esmall{i}$) items in $\rS_i$. 
	
	
	For any $j \in [\vert \tS \vert]$, we define an indicator random variable $Y_j \in \set{0,1}$ which is $1$ iff the $j$-th set chosen in $\tS$ is some $\rX_v$ for a descendent $v$ of $v_{i+1}$ (notice that this event is based on the 
	set of items chosen in $\tS$ \emph{not their realizations}). 
	Let $u_1,\ldots,u_d$ be the $d$ child-vertices of $v_{i}$. We have, 
	\begin{align}
		\Pr_{v_{i+1}}\paren{Y_j = 1 \mid \esmall{1},\ldots,\esmall{i},\ehit{1},\ldots,\ehit{i-1}} \leq \frac{1}{d}. \label{eq:1/d}
	\end{align}
	This is simply because only $1/d$ fraction of descendants of $v_i$ are also descendent of $v_{i+1}$ as $\TT$ is a $d$-ary tree. Define $Y = \sum_{j=1}^{\card{\tS}}Y_j$, i.e., the number of sets chosen from a descendent of $v_{i+1}$: 
	\begin{align*}
		&\Pr\paren{Y \geq 1 \mid \esmall{1},\ldots,\esmall{i},\ehit{1},\ldots,\ehit{i-1}} \\
		&\qquad \leq \Ex\bracket{Y \mid \esmall{1},\ldots,\esmall{i},\ehit{1},\ldots,\ehit{i-1}} \tag{Markov inequality} \\
		&\qquad \!\!\!\!\!\Leq{Eq~(\ref{eq:1/d})} \frac{\vert \tS \vert}{d} \leq \frac{1}{8k}. \tag{as $d = N^2$ and $\vert \tS \vert \leq N/8k$  and $N \geq 1$}
	\end{align*}
	Now notice that under event $\esmall{i}$, in the $i$-th round, we only pick the sets that are in $\tS$ and hence under this conditioning, the probability that any descendants of $v_{i+1}$ belongs to $\tS_i$ is at most $1/8k$. This concludes the proof.
\end{proof}

Define the events $\esmall{*} := \esmall{1},\ldots,\esmall{k-1}$ and $\ehit{*} := \ehit{1},\ldots,\ehit{k-1}$. 
We now prove that conditioned on these two events, expected cost of $\alg$ is large, in particular $\rS_k$ needs to be large in expectation.

\begin{lemma}\label{lem:r-lower-main}
	$\Ex_{S_1,\ldots,S_{k-1}}\Ex_{\rS_k}\bracket{\card{\rS_k} \mid S_1,\ldots,S_{k-1},\esmall{*},\ehit{*}} = \Omega(N/k)$. 
\end{lemma}
\begin{proof}
	Fix any $S_1,\ldots,S_{k-1}$ conditioned on events $\esmall{*},\ehit{*}$; as argued before, these events are only a function $S_1,\ldots,S_{k-1}$. 
	We now bound $\card{\rS_k}$ in expectation.

	Recall that $v_k$ is the $k$-th vertex of the canonical path of $\rX^{(k)}$ which is a leaf vertex of $\TT$. By event $\ehit{*}$, we know that $\rX_{v_{k}}$ has not been chosen by $\alg$ in $S_1,\ldots,S_{k-1}$. 
	As such, conditioned on $\paren{S_1,\ldots,S_{k-1},\esmall{*},\ehit{*}}$, the set $\rX_{v_k}$ still realizes to some set $U_{v_k} \setminus \set{\estar}$ for $\estar \in U_{v_k}$ uniformly at random. In particular, for any element $e \in U_{v_k}$,
	\begin{align}
	\Pr_{\estar}\paren{\estar = e \mid S_1,\ldots,S_{k-1},\esmall{*},\ehit{*}} = \frac{1}{\card{U_{v_k}}}. \label{eq:e*-e}
	\end{align}
	
	Let $\Uc$ be the set of elements covered in the first $k-1$ rounds, i.e., by $S_1,\ldots,S_{k-1}$. 
	Let $U'_{v_k} := U_{v_k} \setminus \Uc$ be the set of elements in $U_{v_k}$ which are \emph{not} covered in the first $k-1$ rounds. As $S_1,\ldots,S_{k-1}$ do not 
	contain $\rX_{v_k}$, we can apply Lemma~\ref{lem:int-small} and obtain that 
	\begin{align}
		\card{U'_{v_k}} &= \card{U_{v_k}} - \card{U_{v_k} \cap \Uc} \\
			 &\Geq{Lemma~\ref{lem:int-small}} \card{U_{v_k}} - \sum_{i=1}^{k-1}\card{S_{i}} \cdot 2k \geq N - (N/8k) \cdot 4k \\
			 &= N/2, \label{eq:u'-v-k} 
	\end{align}
	as by event $\esmall{*}$, $\card{S_i} \leq N/8k$ for all $i \in [k-1]$. 
	
	Conditioned on $S_1,\ldots,S_{k-1}$, the ordering chosen for $\rS_k$ is fixed. Let $\tau := N/16k$ and $\rX_1,\ldots,\rX_{\tau}$ be the first $\tau$ sets in this ordering.
	Now consider the element $\set{\estar} = U_{v_k} \setminus \rX_{v_k}$; this element is chosen uniformly at random from $U_{v_k}$ as argued before. We lower bound the probability that the first 
	$\tau$ sets in $\rS_k$ can cover this element $\estar$. Clearly $\rX_{v_k}$ cannot cover $\estar$, hence in the following, without loss of generality, we assume that $\rX_1,\ldots,\rX_{\tau}$ do not contain $\rX_{v_k}$. 
	This together with Lemma~\ref{lem:int-small} implies that $\card{(\rX_1 \cup \ldots \rX_{k}) \cap U_{v_k}} \leq \tau \cdot 4k$. 
	We have, 
	\begin{align*}
		\Pr\paren{\estar \in \Uc \cup \rX_1 \cup \ldots \cup \rX_{\tau} \mid S_1,\ldots,S_{k-1},\esmall{*},\ehit{*}} &\!\!\Leq{Eq~(\ref{eq:e*-e})}  \frac{\card{\Uc}}{U_{v_k}} + \frac{\card{(\rX_1 \cup \ldots \cup \rX_{k-1}) \cap U_{v_k}}}{\card{U_{v_k}}} \\
		&\!\Leq{Eq~(\ref{eq:u'-v-k})} \frac{N}{2N} + \frac{\tau \cdot 4k}{N} = \frac{3}{4}. \tag{by choice of $\tau = N/16k$ and since $\card{U_{v_k}} = N$ by Property (II) of edifice in Definition~\ref{def:edifice}}
	\end{align*}
	This means that with probability at least $1/4$, $\rS_k$ needs to pick more than $\tau$ sets to cover the universe $U$ (in particular the element $\estar$), hence, 
	\begin{align*}
		\Ex_{\rS_k}\bracket{\card{\rS_k} \mid S_1,\ldots,S_{k-1},\esmall{*},\ehit{*}} \geq \tau/4 = \Omega(N/k).
	\end{align*}
	Taking an expectation over $S_1,\ldots,S_{k-1}$ conditioned on $\esmall{*},\ehit{*}$ concludes the proof. 
\end{proof}

We are now ready to finalize the proof. 
\begin{lemma}\label{lem:r-lower-exp}
	$\Ex_{X \sim \rX^{(k)}}\bracket{\alg(X)} = \Omega(N/k^2)$.
\end{lemma}
\begin{proof}
	We can write the expected cost of $\alg$ as:
	\begin{align*}
		\Ex_{X \sim \rX^{(k)}}\bracket{\alg(X)} &= \Ex_{S_1} \Ex_{X} \Bracket{\alg(X) \mid S_1} \\
		&= \Pr\paren{\esmall{1}} \cdot \Ex_{S_1} \Ex_{X} \Bracket{\alg(X) \mid S_1,\esmall{1}} \\
		&\qquad \qquad + \paren{1-\Pr\paren{\esmall{1}}} \cdot \Ex_{S_1} \Ex_{X} \Bracket{\alg(X) \mid S_1,\overline{\esmall{1}}} \\
		&\geq \Pr\paren{\esmall{1}} \cdot \Ex_{S_1 } \Ex_{X} \Bracket{\alg(X) \mid S_1,\esmall{1}} \\ 
		&\qquad \qquad + \paren{1-\Pr\paren{\esmall{1}}} \cdot N/8k. 
	\end{align*}
	The inequality is by definition of $\overline{\esmall{1}}$ as this means that $\card{S_1} \geq N/8k$. As such, if $\Pr\paren{\esmall{*}} \leq (1-1/2k)$, we are already done as in this case the second term in RHS above
	is at least $(N/8k) \cdot (1/2k) = \Omega(N/k^2)$.
	Otherwise, 
	\begin{align*}
		\Ex_{X \sim \rX^{(k)}}\bracket{\alg(X)} &\geq \paren{1-1/2k} \cdot \Ex_{S_1} \Ex_{X} \Bracket{\alg(X) \mid S_1,\esmall{1}} \\ 
		&\geq \paren{1-1/2k} \cdot \Pr\paren{\ehit{1} \mid \esmall{1}} \Ex_{S_1} \Ex_{X} \Bracket{\alg(X) \mid S_1,\ehit{1},\esmall{1}} \\
		&\!\!\!\!\!\!\!\Geq{Claim~\ref{clm:event-v-i}} \paren{1-1/2k}^{2} \cdot \Ex_{S_1} \Ex_{X} \Bracket{\alg(X) \mid S_1,\ehit{1},\esmall{1}}. 
	\end{align*}
	
	We now continue this calculation for the RHS using the sets $S_2$ in second round: 
	\begin{align*}
		 &\Ex_{S_1} \Ex_{X} \Bracket{\alg(X) \mid S_1,\ehit{1},\esmall{1}} \\
		 &\qquad = \Ex_{S_1} \Ex_{S_2}\Ex_{X} \Bracket{\alg(X) \mid S_2,S_1,\ehit{1},\esmall{1}} \\
		&\qquad = \Pr\paren{\esmall{2} \mid \ehit{1},\esmall{1}}\Ex_{S_1} \Ex_{S_2}\Ex_{X} \Bracket{\alg(X) \mid S_2,S_1,\esmall{2},\ehit{1},\esmall{1}} \\
		&\qquad \qquad + \Pr\paren{\overline{\esmall{2}} \mid \ehit{1},\esmall{1}} \cdot \Ex_{S_1} \Ex_{S_2}\Ex_{X} \Bracket{\alg(X) \mid S_2,S_1,\overline{\esmall{2}},\ehit{1},\esmall{1}}	
	\end{align*}
	Again, if $\Pr\paren{\esmall{2} \mid \ehit{1},\esmall{1}} \leq (1-1/2k)$, we are already done as in this case the second term in RHS above is at least $\Omega(N/k^2)$. Combining this with previous equation, we obtain 
	that expected cost of $\alg$ is at least $(1-1/2k)^3 \cdot \Omega(N/k^2) = \Omega(N/k^2)$. Hence, we can assume that $\Pr\paren{\esmall{2} \mid \ehit{1},\esmall{1}} \geq (1-1/2k)$. Using this, and the previous argument we did for the first round, 
	and by Claim~\ref{clm:event-v-i}, we obtain that:
	\begin{align*}
		\Ex_{X \sim \rX^{(k)}}\bracket{\alg{(X)}} \geq \paren{1-\frac{1}{2k}}^4 \cdot \Ex_{S_1} \Ex_{S_2}\Ex_{X} \Bracket{\alg(X) \mid S_2,S_1,\ehit{2},\esmall{2},\ehit{1},\esmall{1}}.
	\end{align*}
	We can thus continue this argument until processing the last round, and either we already have $\Ex_{X \sim \rX^{(k)}} = \Omega(N/k^2)$ as for some $i \in [k-1]$, $\Pr\paren{\esmall{i} \mid \small{1},\ldots,\esmall{i-1},\ehit{1},\ldots,\ehit{i-1}} \geq (1-1/2k)$, 
	or: 
	\begin{align*}
		\Ex_{X \sim \rX^{(k)}}\bracket{\alg{(X)}} &\geq \paren{1-\frac{1}{2k}}^{2k-2} \cdot \Ex_{S_1,\ldots,S_{k-1}}\Ex_{X} \Bracket{\alg(X) \mid S_1,\ldots,S_{k-1},\ehit{*},\esmall{*}} \\
		&\geq \Omega(1) \cdot  \Ex_{S_1,\ldots,S_{k-1}}\Ex_{\rS_k}\bracket{\card{\rS_k} \mid S_1,\ldots,S_{k-1},\ehit{*},\esmall{*}} \\
		&\!\!\!\!\!\!\!\!\!\Geq{Lemma~\ref{lem:r-lower-main}} \Omega(N/k).
	\end{align*}
	This concludes the proof. 
\end{proof}

Theorem~\ref{thm:r-round-lower} now follows from Lemma~\ref{lem:r-lower-exp} and Claim~\ref{clm:opt-adaptive-lower}, by setting $r = k$ and noticing that $N = n^{1/k}$ in this construction.